\documentclass[12pt,onecolumn,draftclsnofoot]{IEEEtran}

\usepackage{color}

\usepackage{amsmath}
\usepackage{multirow}
\usepackage{amssymb}

\usepackage{algorithm}
\usepackage[noend]{algpseudocode}

 \usepackage{amsthm}
 
\newtheorem{corollary}{Corollary}
\newtheorem{proposition}{Proposition}
\newtheorem{lemma}{Lemma}

\usepackage[font=small,labelfont=bf]{caption}

\usepackage{subcaption}
 \usepackage{graphicx}

\usepackage{cite}
\ifCLASSINFOpdf
\else
\fi

\begin{document}
%
\title{Transmission Scheme, Detection and Power Allocation for Uplink User Cooperation with NOMA and RSMA}
\author{Omid~Abbasi, \IEEEmembership{Student Member,~IEEE}, Halim~Yanikomeroglu, \IEEEmembership{Fellow,~IEEE} \thanks{This work was supported by Huawei Canada Co.,
Ltd.}
\thanks{O. Abbasi and H. Yanikomeroglu are with the Department of Systems and Computer
Engineering, Carleton University, Ottawa, ON K1S5B6, Canada. e-mail: omidabbasi@sce.carleton.ca; halim@sce.carleton.ca (\textit{Corresponding author: Omid Abbasi})}
\thanks{A preliminary version of this work appeared in the Proceedings of the 2021 IEEE Wireless Communications and Networking Conference (WCNC) \cite{conf_version}.}}
\maketitle

\begin{abstract}
  In this paper, we propose two novel cooperative-non-orthogonal-multiple-access (C-NOMA) and cooperative-rate-splitting-multiple-access (C-RSMA) schemes for uplink user cooperation. At the first mini-slot of these schemes, each user transmits its signal and receives the transmitted signal of the other user in full-duplex mode, and at the second mini-slot, each user relays the other user's message with amplify-and-forward (AF) protocol. At both schemes, to achieve better spectral efficiency, users transmit signals in the non-orthogonal mode in both mini-slots. In C-RSMA, we also apply the rate-splitting method in which the message of each user is divided into two streams.
  In the proposed detection schemes for C-NOMA and C-RSMA, we apply a combination of maximum-ratio-combining (MRC) and successive-interference-cancellation (SIC). Then, we derive the achievable rates for C-NOMA and C-RSMA, and formulate two optimization problems to maximize the minimum rate of two users by considering the proportional fairness coefficient. We propose two power allocation algorithms based on successive-convex-approximation (SCA) and geometric-programming (GP) to solve these non-convex problems.  Next, we derive the asymptotic outage probability of the proposed C-NOMA and C-RSMA schemes, and prove that they achieve diversity order of two. Finally, the above-mentioned performance is confirmed by simulations.
\end{abstract}


%
\IEEEpeerreviewmaketitle

\section{Introduction}
\label{intro}
Three generic
services of the fifth-generation (5G) of wireless cellular systems are enhanced mobile broadband (eMBB), ultra-reliable low-latency communications (URLLC),
and massive machine-type communications (mMTC) \cite{NR_magazine}. eMBB accommodates high
data rate applications such as 4K video and virtual reality
(VR), and it can leverage coding
over large transmission blocks due to its non-critical latency requirements. In contrast, URLLC
imposes strict latency constraints such as autonomous vehicles, tactile internet, or
remote surgery, hence requiring transmissions
localized in time, while still ensuring high-reliability levels. In general, URLLC transmissions are sporadic
with a short packet size and with a relatively low data rate. mMTC traffic consists of a large
number of devices transmitting small payloads to a common receiver. New technological
trends will redefine 5G application types by combining
classical URLLC, eMBB, and mMTC and introducing new
services such as eMBB\&URLLC \cite{6G_saad}.
The distinction between
eMBB and URLLC will no longer be sustainable to support
applications such as extended reality (including augmented reality, VR), connected autonomous vehicles, autonomous drones, and blockchain. This is
because these applications require, not only high reliability and
low latency but also high data rates. \par

Cooperative communication has been well established in wireless networks that enable users to share their antennas and generate a virtual multiple-antenna transmitter that allows them to achieve transmit diversity \cite{Nosratinia}. Authors in \cite{Aazhang} proposed a new form of spatial diversity, in which diversity gains
were achieved via the cooperation of mobile users. They showed that, even
though the inter-user channel is noisy, cooperation leads not only
to an increase in capacity for both users but also to a more robust
system, where users’ achievable rates are less susceptible to
channel variations. 
Authors in \cite{kaya_power} obtained the power allocation policies for the cooperative scheme in \cite{Aazhang}
that maximize the average rates for a fading Gaussian multiple access channel (MAC), subject to average power constraints. In \cite{kaya_ofdma}, the authors proposed encoding strategies for a two
user cooperative orthogonal frequency division multiple access (OFDMA) system and obtained the expressions for the resulting
achievable rate regions. In \cite{laneman_04}, the authors developed low-complexity cooperative
diversity protocols that combat fading induced by multipath
propagation in wireless networks. They analyzed the outage probability and diversity order of their proposed schemes. They proved that
using distributed antennas, one can reach the powerful benefits
of space diversity without the need for physical arrays, though at a loss
of spectral efficiency due to half-duplex operation and at
the cost of additional receiver hardware.
Authors in \cite{Nabar_04} studied three different time division multiple access  (TDMA)-based cooperative protocols
for a simple fading relay channel. For each of the protocols, assuming Gaussian codebooks, they derived the ergodic and outage capacities. In \cite{superposition_modulation_larsson_05}, the authors proposed a new strategy for cooperative transmit
diversity based on superposition modulation and multiuser detection.
Their new scheme outperformed classical cooperative
diversity based on decode-and-forward (DF) at the same computational complexity. Authors in \cite{azarian05} proposed a new amplify and forward (AF) protocol for cooperative MAC that achieves the optimal
diversity-multiplexing tradeoff \cite{tradeoff}. A distinguishing feature of their
proposed protocol was that it did not rely on
orthogonal subspaces, allowing for more efficient use of resources.
Their results showed that the suboptimality of
previously proposed cooperative protocols stems from their use of orthogonal
subspaces rather than the half-duplex constraint. In \cite{performace_11}, the authors proposed a low-complexity suboptimal detector for a two-user cooperative MAC that utilizes the superposition modulation
scheme proposed in \cite{superposition_modulation_larsson_05}. Bit error
rate results showed that the superposition modulation can achieve both full diversity order and high coding
gain. \par
Non-orthogonal multiple access (NOMA) has attracted tremendous attention from academia and industry in recent years due to its higher spectral efficiency rather than orthogonal multiple access (OMA) \cite{saito2013non}. In NOMA, the signals of different users are differentiated in the power domain, and users transmit their signals in the same time/frequency/code resources \cite{uplink_noma_power}. It has been shown that in  MAC, transmitting signals at the same resources and decoding them at the receiver by performing successive interference cancellation (SIC) achieve the capacity of MAC in terms of sum rate \cite{tse2005fundamentals}. Note that transmitting signals of users in the same resources and performing SIC at the receiver side without time sharing is called uplink NOMA, and it can not achieve all the points of the capacity region in MAC.
 Indeed, the capacity region of MAC can be achieved by performing SIC and time sharing for different decoding orders of users' signals \cite{tse2005fundamentals}. In \cite{Han-kobayashi}, it has been shown that Han-Kobayashi coding scheme can achieve all the points of capacity region for the two-user interference channel. Authors in \cite{capacity_rsma} showed that rate-splitting (RS) without performing time sharing can achieve the capacity region of MAC. They showed that by splitting the signals of each user into two streams and forming two virtual sources for each user at the transmitter side, and performing SICs for received symbols at the receiver side, the capacity region can be achieved.\par

\subsection{Related Works}
Rate-splitting multiple access (RSMA) has been
recently proposed as an effective approach to provide a more
general and robust transmission framework compared to NOMA and space-division
multiple access (SDMA) \cite{Clercks}, \cite{Unified_RSMA}. Most of the existing works studied the use of RSMA for the downlink rather than in the uplink. On the other hand, using RSMA
for uplink data transmission can theoretically achieve the optimal rate region \cite{capacity_rsma}. In \cite{RSMA_uplink}, the authors proposed an uplink RSMA
scheme to avoid user paring
and significantly reduce the scheduling complexity. They derived
the exact closed-form expressions of the outage probability and
achievable sum rate of the proposed scheme. In \cite{RS_code}, authors incorporated RS into the
design of an uplink code domain NOMA scheme. By utilizing SIC, the code domain NOMA with
RS reduced the decoding complexity dramatically while achieving
similar performance as those taking the belief propagation
algorithm. In \cite{Baluano}, an RS
scheme to guarantee max-min fairness of uplink single-input-multiple-output (SIMO)-NOMA
was studied. In \cite{Mingzhi_journal,saad_uplink}, the problem of maximizing the wireless users’ sum-rate for uplink RSMA communications was studied. They solved a sum-rate maximization
problem with proportional rate constraints by adjusting the users’ transmit power and the BS’s decoding
order. Authors in \cite{poor_rsma}, proposed two kinds
of RS schemes, namely, fixed RS and cognitive RS
schemes, to realize RS for uplink NOMA with the
aim of improving user fairness and outage performance in delay-limited
transmissions.\par
Cooperation among users has been utilized a lot in downlink NOMA-based communication. In \cite{C-NOMA_downlink}, the authors proposed a cooperative NOMA transmission
scheme that uses the fact that strong users in NOMA
systems have prior information about the weak users' messages. In \cite{Full_CNOMA}, the authors proposed a full-duplex device-to-device (D2D)-aided cooperative NOMA scheme to improve
the outage performance of the weak user in a NOMA user
pair, where the weak user is helped by the strong user with
the capability of full-duplex D2D communications. Authors in \cite{occupy_COV} proposed a cooperative wireless communication protocol
framework for high-performance industrial-automation systems
that demand ultra-high reliability and low latency for
many message streams within a network with many active
nodes. They proposed to combine the user
messages as a single packet, then transmit the packet using
a two-phase relaying strategy in order to harvest diversity. Authors in \cite{Wei_Yu} devised an alternative
approach that splits the per-cell message into the broadcast part
and the relay part, thereby enabling layered data transmissions
to the receivers of various channel conditions. In \cite{inspired_omid}, we proposed a cooperative relaying system (CRS) based on RS with an amplify-and-forward relay. We introduced a novel detection scheme in which the received signals of two consecutive time slots at the destination were applied. Authors in \cite{cooperative_RS_clercks} proposed a cooperative RS for multiple-input-single-output (MISO) broadcast channel with user relaying. In \cite{Traj_Omid}, we showed that in a two-user network with a dedicated AF relay, NOMA always has a better or equal sum-rate in comparison to OMA at a high signal-to-noise-ratio (SNR) regime. In \cite{User_relaying_Mao}, authors investigated the max-min fairness of
K-user CRS for MISO channel with user relaying by designing
the precoder, the RS message split and time slot allocation. Authors in \cite{secrecy_mao} employed the cooperative RS technique to enhance the secrecy sum rate for the MISO Broadcast Channel, consisting of two legitimate users and one eavesdropper. All of the above-mentioned works consider non-orthogonal cooperative communication for downlink. In \cite{Dedicated_AF}, the authors proposed an uplink cooperative NOMA system, where a dedicated AF relay was used to help two uplink users transmit information to the base station (BS). Authors in \cite{Dedicated_DF} proposed an uplink cooperative NOMA
system, where a dedicated full-duplex decode and forward relay is used to help two uplink users. In \cite{IET_uplink_CNOMA}, the authors proposed an uplink NOMA system with cooperative
full-duplex relaying, where the user closer to the BS is considered as a full-duplex relay to aid
the transmission from the far user to the BS. The closer user first decodes the signals transmitted from the far user and then
forwards them using superposition coding (SC) to the BS on top of transmitting its own information signals to the BS.

\subsection{Motivations and Contributions}
   
 The cooperative scheme with AF relaying in \cite{laneman_04} was performed at the orthogonal resource blocks, and we name it cooperative (C)-OMA. In this paper, in order to increase spectral efficiency, we propose a novel cooperative non-orthogonal transmission and detection scheme, and find its performance gain over the C-OMA scheme. We name this proposed method as cooperative (C)-NOMA scheme. At the first mini-slot of the proposed C-NOMA scheme, each user transmits its own signal in non-orthogonal mode and at the same resource block. We assume that users work in full-duplex mode, and hence receive the message of another user at the first mini-slot. At the second mini-slot, users relay the other user's message with AF protocol \footnote{At AF relaying each user does not require to decode the other user's message, and hence it is better for low latency communication \cite{patel2006statistical}. On the other hand, some users may not be interested in their messages' decoding by the other users, and hence AF relaying is better than DF relaying from a secrecy perspective.} again at the same resource block. Then we propose a detection scheme at the BS to decode the received signals from two mini-slots at the C-NOMA scheme. We apply maximum ratio combining (MRC) in order to combine the received signals from two mini-slots and attain full diversity. We also apply SIC in order to remove the interference of decoded symbols and achieve better spectral efficiency.\par
 The second novel transmission and detection scheme that we propose in this paper is based on the rate-splitting method. We name this proposed method as cooperative (C)-RSMA scheme, and find its performance gain over traditional C-OMA and the proposed C-NOMA schemes in this paper. In the transmission phase, the message of each user is divided into two streams. Then a combination of these two streams at each user is transmitted and relayed in non-orthogonal mode at the first and second mini-slots, respectively. We propose a detection scheme at the BS for this C-RSMA scheme in which the streams of users are decoded alternately. Similar to the C-NOMA scheme, the MRC and SIC are applied to decode the streams of each user in the proposed detection scheme. Finally, the original messages of each user are reconstructed at the BS by combining the two decoded streams for that user.  \par

After proposing two novel transmission and detection schemes, we derive the achievable rates corresponding to these proposed C-NOMA and C-RSMA schemes. Then, we formulate two optimization problems maximizing the minimum rate of two users for C-NOMA and C-RSMA schemes by considering the proportional fairness coefficient. These two optimization problems are non-convex, and hence we propose two low complexity power allocation algorithms to solve them. In these algorithms, we firstly apply the successive convex approximation (SCA) method to transform the original non-convex problems into a geometric programming (GP) problem. The GP problems are well-known problems \cite{boyd2004convex} that can be transformed into convex problems. We prove that the proposed efficient algorithms are guaranteed to converge. Finally, we derive the asymptotic outage probability of the proposed C-NOMA and C-RSMA schemes. Also, we prove that the proposed cooperative schemes achieve  a diversity order of two. Simulation results prove the superiority of the proposed C-NOMA and C-RSMA schemes rather than C-OMA and non-cooperative schemes in terms of achievable rate and outage probability.\par 
In this paper, we propose two novel transmission and detection schemes for C-NOMA and C-RSMA that can be completed in two mini-slots. Note that in the air interface of 5G new-radio (NR), the mini-slots have been considered for URLLC applications \cite{NR_magazine}. Indeed, a mini-slot consists of a smaller number of symbols compared with a slot, and therefore they cause a lower latency. 
Note that for URLLC scenarios, both low latency and high-reliability communication is required. In this paper, in order to reach high reliability and increase diversity, we propose two novel detection methods for the C-NOMA and C-RSMA schemes. 
 \par
 Our system model can be generalized into a multi-user system with $K$ users by dividing the $K$ users into $\frac{K}{2}$ groups with two users inside each group. Then we can easily apply the proposed transmission and detection schemes for the C-NOMA and C-RSMA schemes for each group. In order to generalize our formulated optimization problems and their solutions so that they can be applied into this multi-user case, we just need to define user paring coefficients \cite{Mingzhi_journal,uplink_pairing}, and pair users based on the channel quality between them. Then, the proposed algorithms for power allocations in the manuscript can be applied into each group. Note that these groups must be distinguished by different time/frequency resource blocks or different beams. In order to serve each pair by one beam, we can consider that the BS is equipped with  $M>\frac{K}{2}$ antennas and creates  $\frac{K}{2}$ beams. Each of these beams is responsible to serve one pair of users, and within each of these $\frac{K}{2}$ beams, the proposed C-NOMA and C-RSMA schemes can be applied. \\
 The main contributions of this paper are summarized as follows:
\begin{itemize}
 \item \textbf{C-NOMA scheme:} In order to increase spectral efficiency, we propose a novel cooperative non-orthogonal transmission and detection scheme. At the first mini-slot of the proposed C-NOMA scheme, each user transmits its own signal in non-orthogonal mode, and at the second mini-slot, users relay the other user's message with AF protocol again at the same resource block. Then we propose a detection scheme based on MRC and SIC at the BS to decode the received signals from two mini-slots.
    \item \textbf{C-RSMA scheme:} In this rate-splitting-based scheme, the message of each user is divided into two streams. Then, a combination of these two streams at each user is transmitted and relayed in non-orthogonal mode at the first and second mini-slots, respectively. We propose a detection scheme at the BS for this C-RSMA scheme in which the streams of users are decoded alternately. The MRC and SIC are applied to decode the streams of each user in the proposed detection scheme. Finally, the original messages of each user are reconstructed at the BS by combining the two decoded streams for that user.
    
    \item \textbf{Achievable rates and Power allocation:} We derive the achievable rates corresponding to the proposed cooperative schemes. Then, we formulate two optimization problems in order to maximize the minimum rate of two users for C-NOMA and C-RSMA schemes by considering the proportional fairness coefficient. We propose two low complexity power allocation algorithms based on SCA and GP to solve these non-convex problems. We prove that the proposed efficient algorithms are guaranteed to converge.
    \item \textbf{Outage probability and diversity order:} We derive the asymptotic outage probability of the proposed C-NOMA and C-RSMA schemes, and we prove that the proposed cooperative schemes achieve  a diversity order of two.
\end{itemize}

\subsection{Organization}
The remainder of this paper is organized as follows. Section II presents the system model. Section III presents the proposed detection schemes and derived achievable rates for C-NOMA and C-RSMA. Section IV provides the formulated optimization problems for C-NOMA and C-RSMA schemes. In Section V, we derive the outage probability and diversity order of the proposed schemes. Section VI provides simulation
results to validate the performance of the proposed schemes. Finally, Section VII concludes the paper.

%
%
%
%




\section{System Model}
We consider an uplink communication in which two users 
cooperate to transmit their messages to the BS as depicted in Fig. \ref{system-model}. At the first mini-slot, each user transmits its message to the BS. We assume that there is a strong link between two users so that two users can receive the message from each other in full-duplex mode. At the second mini-slot, each user relays the received message of the other user in the previous min-slot with AF protocol. Note that in AF relaying, each user just amplifies and retransmits the received signal of the other user. The BS, which is equipped with a single antenna, utilizes the received signals of two users at two mini-slots and decodes their messages. We explain the proposed transmission and detection schemes for the users and BS in the next section. \par
In this paper, the channel power gains of user 1 to the BS and user 2 to the BS are indicated by $h_1$ and $h_2$, respectively. Also, we assume that the channel reciprocity holds for the links between user 1 and user 2, and we show the channel power gain of these links by $h_3$. Note that we investigate the performance of our schemes in the range of the coherence time and coherence bandwidth interval, and hence we consider the same channel power gains for two mini-slots and different subcarriers. Also, note that $h_k$ for $k\in \{1, 2, 3\}$ includes both the path loss and multipath fading, and hence we have $h_k=g_k\beta_0d_k^{-\alpha}$ for $k\in \{1, 2, 3\}$. In this equation, $g_k$ stands for the multipath fading component of channel power gain which has an exponential distribution with the mean value of $\lambda_{g_k}=1$. Also, $\beta_0d_k^{-\alpha}$ indicates the path loss component of channel power gain in which $d$ indicates the distance between nodes, $\alpha$ is the path loss exponent, and $\beta_{0}$ denotes the channel power at the reference distance $d_{0}=1~\mathrm{m}$. We consider independent additive white Gaussian noise (AWGN) with the distribution $CN(0,\sigma_{k}^{2})$ for $k \in\{1, 2, BS\}$ in which $\sigma_{k}^{2}$ shows the variance of the noise for the node $k$.  $\bar{P}_{k}$ for $k \in\{1, 2\}$  indicates the average transmit power at each mini-slot for each user $k \in\{1, 2\}$.\par


\begin{figure}[t]
\centering
\begin{minipage}{.44\linewidth}
\vspace{1.2cm}
  \includegraphics[width=\linewidth]{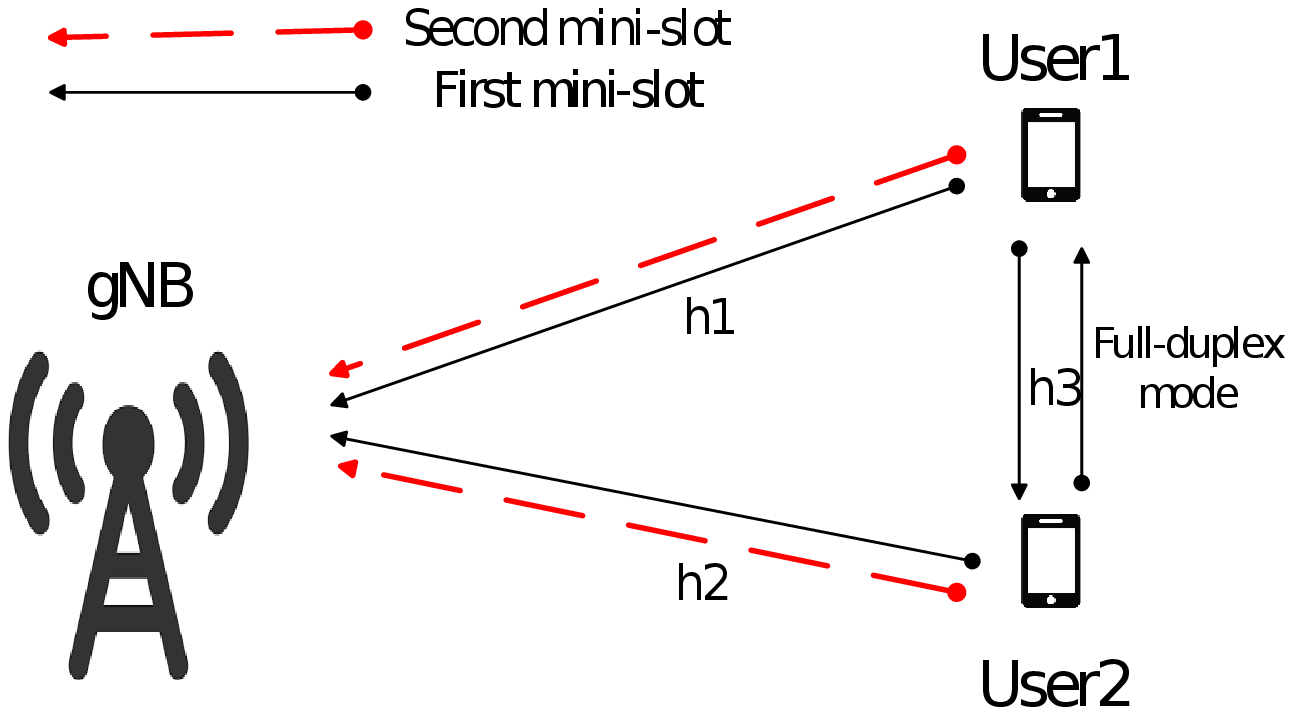}
  \captionof{figure}{System model for proposed user cooperation scheme at uplink. At the first mini-slot, each user transmits its signal and receives the transmitted signal of the other user in full-duplex mode. At the second  mini-slot, each user relays the other user's message with AF protocol. }
  \label{system-model}
\end{minipage}
 \hspace{.02\linewidth}
\begin{minipage}{.52\linewidth}
  \includegraphics[width=\linewidth]{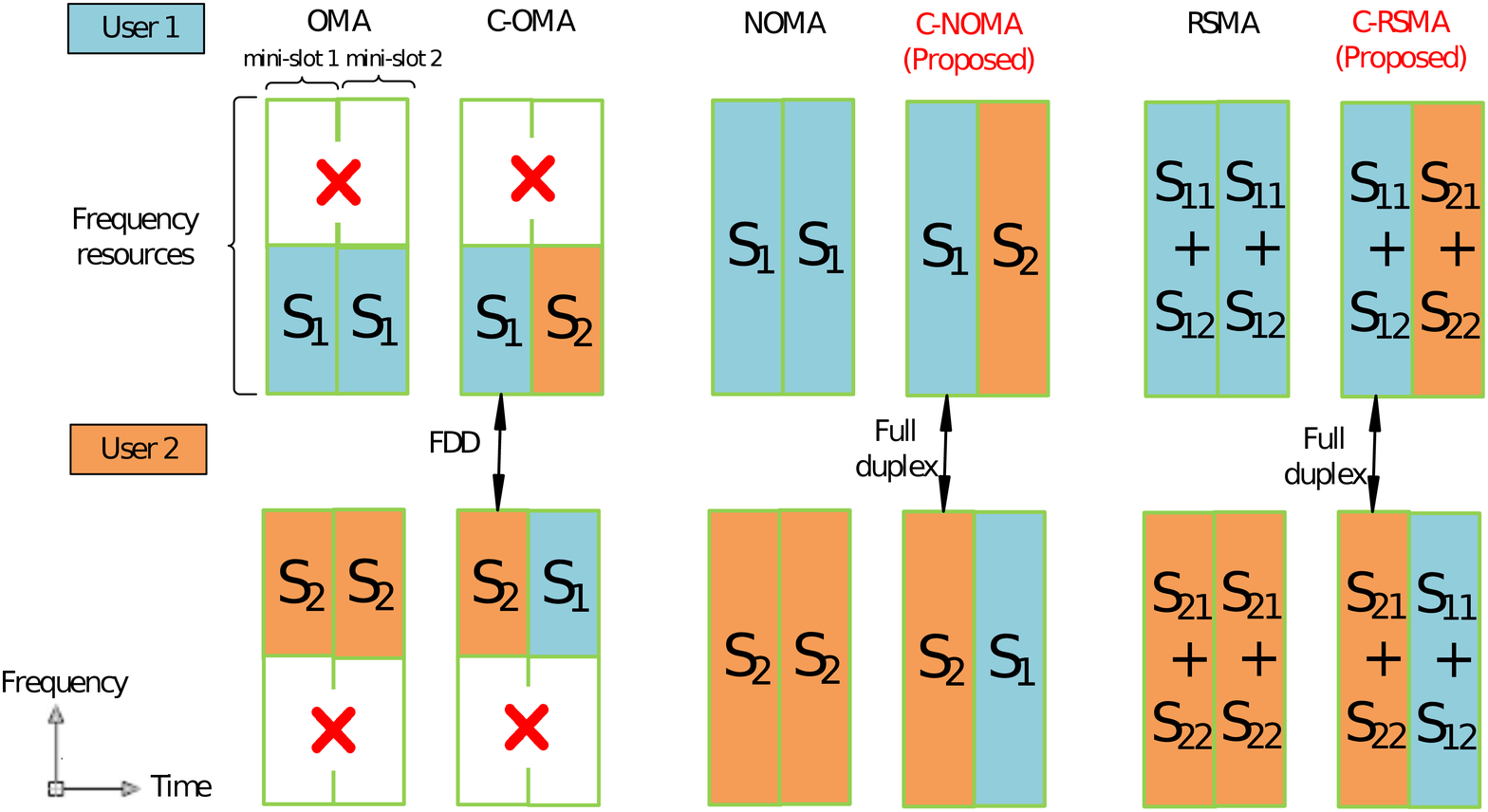}
  \captionof{figure}{Conventional and the proposed transmission schemes at each user with OMA, C-OMA, NOMA, C-NOMA, RSMA, and C-RSMA. In the first mini-slot of C-NOMA and C-RSMA schemes, the communication is performed in the full-duplex mode.}
  \label{frame}
\end{minipage}
\end{figure}

\section{Achievable Rates for C-NOMA and C-RSMA}
In this section, we propose transmission and detection schemes for two cooperative schemes and derive the achievable rates for them.
We can see the transmission scheme at each user for OMA, C-OMA, NOMA, C-NOMA, RSMA, and C-RSMA schemes in Fig. \ref{frame}. In this figure, horizontal and vertical axes indicate time and frequency resources, respectively. Also, for the schemes that do not apply RS, i.e., OMA, C-OMA, NOMA, C-NOMA, $s_i$ shows the symbol of user $i$ for $i\in\{1, 2\}$. For RSMA and C-RSMA schemes, $s_{i1}$ and $s_{i2}$ for $i\in\{1, 2\}$ show two streams of the user $i$ after splitting. Note that for cooperative schemes, we need to complete the whole transmission in two mini-slots. On the other hand, for orthogonal schemes, two frequency resource blocks are required for two users to send their data. Therefore, as one can see in Fig. \ref{frame}, in order to have a fair comparison among different multiple access schemes, we consider two mini-slots and two frequency resource blocks (four resource blocks in total) for all of the schemes. This means that the achievable rates for all schemes are for the same number of resource blocks, and all schemes have the same latency of two mini-slots. \par
In C-OMA, at the first mini-slot, each user transmits its own message to the BS in the orthogonal frequency resources. This means that the communication between two users is performed in the frequency division duplexing (FDD) mode at the C-OMA scheme. However, in the proposed transmission schemes for C-NOMA and C-RSMA, at the first mini-slot, two users transmit their own messages in the same frequency resources. This means that the communication between two users at the first mini-slot is performed in the full-duplex mode at the C-NOMA and C-RSMA schemes, and hence each user receives the message of another user at the first mini-slot. In C-NOMA and C-RSMA schemes, at the second mini-slot, users relay the other user's message with AF protocol again at the same resource block. In this paper, without loss of generality, we assume that user 1 has a better channel condition than user 2, i.e., $h_1>h_2$. Also, the notation $i'$ for the index of users means the other user than user $i$.

\subsection{C-NOMA}
\par In the following proposition, the achievable rates of each user with the C-NOMA scheme are derived.
\begin{proposition}\label{proposition-achievable-rate-C-NOMA}
In a two-user MAC with a strong link between two users, applying C-NOMA leads to the following achievable rates for each user
\small
\begin{equation}\label{r1_cnoma}
  R_1^{\mathrm{C-NOMA}}=\frac{1}{2}\log_2(1+\frac{\gamma_1p_1^1}{\gamma_2p_2^1+1}+\frac{\gamma_2p_2^2}{\gamma_1p_1^2+1}),  
\end{equation}
and
\begin{equation}\label{r2_cnoma}
 R_2^{\mathrm{C-NOMA}}=\frac{1}{2}\log_2(1+\gamma_1p_1^2+\gamma_2p_2^1),  
\end{equation}
\normalsize
where $\gamma_i=\frac{h_i\bar{P}_{i}}{\sigma_{BS}^2}$ for $i\in\{1, 2\}$ is the SNR of each user, and $p_{i}^j$ for $i\in\{1, 2\}$ and $j\in\{1, 2\}$ shows the allocated power for user $i$ at mini-slot $j$.
\end{proposition}
\begin{proof}
In the proposed transmission scheme for C-NOMA (Fig. \ref{frame}), at the first mini-slot, two users transmit their own messages in the same frequency resources. The received signal at the BS in the first mini-slot equals $y_{BS}^1=\sqrt{h_1P_1^1}s_1+\sqrt{h_2P_2^1}s_2+z_{BS}$ in which $P_{i}^j$ (for $i\in\{1, 2\}$ and $j\in\{1, 2\}$) shows the allocated power for user $i$ at mini-slot $j$, $s_i$ ($E\{|s_{i}|^{2}\}=1$) is the transmitted symbol from user $i$, and $z_{BS}$ is the AWGN noise at the BS. At the same time, each user $i$ receives the transmitted signal of the other user $i'$ as $y_i=\sqrt{h_3P_{i'}^1}s_{i'}+z_i$ for $i\in\{1, 2\}$. Note that each user is aware of its own transmitted signal, and hence we assume that each user can apply perfect self-interference cancellation. At the second mini-slot, each user amplifies the received signal from the other user with the amplification gain $\rho_i=\frac{P_i^2}{h_3P_{i'}^1+\sigma_i^2}$ \cite{patel2006statistical}. Then, each user forwards the amplified signal $\sqrt{\rho_i}y_{i}$ to the BS in the same frequency resources. Assuming a strong link between two users, the received signal at the BS in the second mini-slot is given by $y_{BS}^2=\sqrt{h_1\rho_1}y_{1}+\sqrt{h_2\rho_2}y_{2}+z_{BS}=\sqrt{h_1P_1^2}s_2+\sqrt{h_2P_2^2}s_1+z_{BS}$. Finally, the proposed detection scheme in Fig. \ref{detection-scheme-C-NOMA} is applied in order to decode the messages of each user according to the received signals in two time-slots at the BS. At this detection scheme, first, the MRC is applied for the received signals at two mini-slot in order to detect the message of user 1. The achievable rate of this user is given by $R_1^{\mathrm{C-NOMA}}=\frac{1}{2}\log_2(1+\frac{h_1P_1^1}{h_2P_2^1+\sigma_{BS}^2}+\frac{h_2P_2^2}{h_1P_1^2+\sigma_{BS}^2})$. Then, the interference of user 1 is subtracted from the received signals at two mini-slots by applying SIC. Finally, we again apply the MRC for these two new signals at two mini-slot in order to detect the message of user 2. The achievable rate of this user is given by $R_2^{\mathrm{C-NOMA}}=\frac{1}{2}\log_2(1+\frac{h_1P_1^2}{\sigma_{BS}^2}+\frac{h_2P_2^1}{\sigma_{BS}^2})$. Note that the average transmit power at each user must satisfy $P_i^1+P_i^2=\bar{P}_{i}$ for $i\in\{1, 2\}$. By defining the normalized powers at each user as $p_i^1=\frac{P_i^1}{\bar{P}_{i}}$ and $p_i^2=\frac{P_i^2}{\bar{P}_{i}}$ and the SNR for each user as $\gamma_i=\frac{h_i\bar{P}_{i}}{\sigma_{BS}^2}$, the achievable rates for each user can be derived as (\ref{r1_cnoma}) and (\ref{r2_cnoma}), and the proof is completed. 
\end{proof}

\begin{figure}[t]
\centering
\begin{minipage}{.47\linewidth}
\vspace{0.4cm}
  \includegraphics[width=\linewidth]{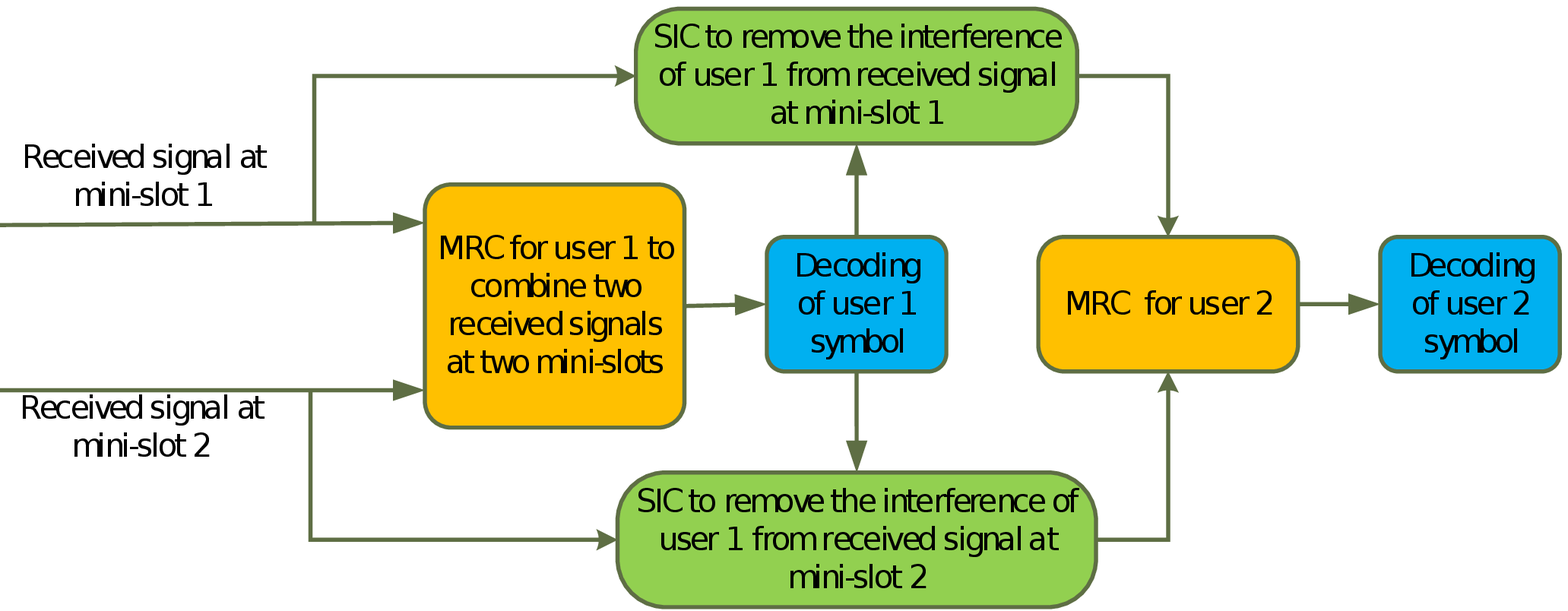}
  \captionof{figure}{Proposed detection scheme for the C-NOMA scheme to decode the messages of two users according to the received signals in two mini-slots at the BS.}
  \label{detection-scheme-C-NOMA}
\end{minipage}
 \hspace{.02\linewidth}
\begin{minipage}{0.47\linewidth}
  \includegraphics[width=\linewidth]{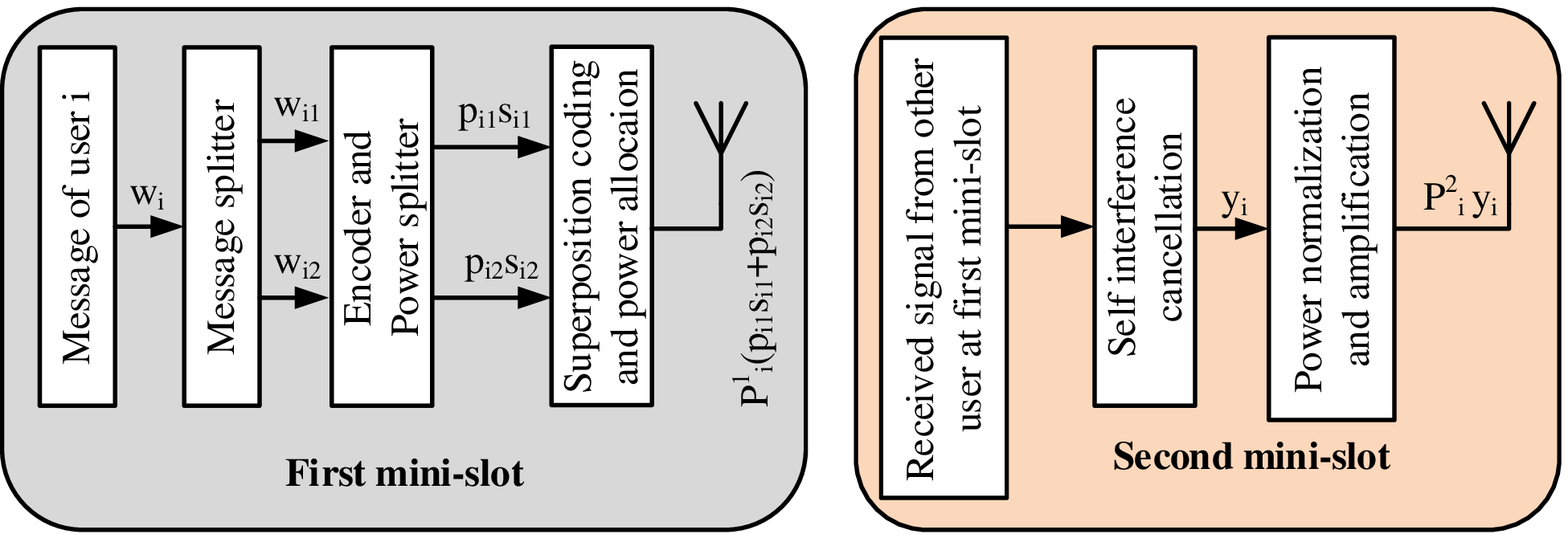}
  \captionof{figure}{Proposed transmission structure at user $i$ for the C-RSMA scheme in mini-slot 1 and mini-slot 2.}
  \label{CRSMA_mini_slots}
\end{minipage}
\end{figure}

\subsection{C-RSMA}
 \par In the following proposition, the achievable rate of each user with the C-RSMA scheme is derived.
\begin{proposition}\label{proposition-achievable-rate-C-RSMA}
In a two-user MAC with a strong link between two users, applying the C-RSMA scheme leads to the following achievable rates for each user
\small
\begin{equation}\label{r1_crsma}
\begin{split}
     R_1^{\mathrm{C-RSMA}}&=\frac{1}{2}\log_2(1+\frac{p_1^1p_{11}\gamma_1}{p_1^1p_{12}\gamma_1+p_2^1p_{21}\gamma_2+p_2^1p_{22}\gamma_2+1}+\frac{p_2^2p_{11}\gamma_2}{p_2^2p_{12}\gamma_2+p_1^2p_{21}\gamma_1+p_1^2p_{22}\gamma_1+1})\\&+\frac{1}{2}\log_2(1+\frac{p_1^1p_{12}\gamma_1}{p_2^1p_{22}\gamma_2+1}+\frac{p_2^2p_{12}\gamma_2}{p_1^2p_{22}\gamma_1+1}),  
\end{split}
\end{equation}
 and
 \begin{equation}\label{r2_crsma}
 \begin{split}
     R_2^{\mathrm{C-RSMA}}&=\frac{1}{2}\log_2(1+\frac{p_2^1p_{21}\gamma_2}{p_2^1p_{22}\gamma_2+p_1^1p_{12}\gamma_1+1}+\frac{p_1^2p_{21}\gamma_1}{p_1^2p_{22}\gamma_1+p_2^2p_{12}\gamma_2+1})\\&+\frac{1}{2}\log_2(1+p_2^1p_{22}\gamma_2+p_1^2p_{22}\gamma_1),
 \end{split}
 \end{equation}
 \normalsize
where $\gamma_i=\frac{h_i\bar{P}_{i}}{\sigma_{BS}^2}$ for $i\in\{1, 2\}$ is the SNR of each user, and $p_{i}^j$ for $i\in\{1, 2\}$ and $j\in\{1, 2\}$ shows the allocated power for user $i$ at mini-slot $j$. Also, $p_{i1}$ and $p_{i2}$ indicate the allocated powers for the  symbol 1 and symbol 2 of user $i$, respectively.
\end{proposition}
\begin{proof}

In the  proposed transmission scheme for C-RSMA in Fig. \ref{CRSMA_mini_slots}, at the first mini-slot, two users transmit their messages in two streams and in the same frequency resources. Note that in the rate-splitting method, each user $i$ divides its message $w_i$ into two messages $w_{i1}$ and $w_{i2}$ for $i\in\{1, 2\}$. These messages then are encoded into two streams $s_{i1}$ and $s_{i2}$ for $i\in\{1, 2\}$ with split powers $p_{i1}$ and $p_{i2}$, respectively. Note that $p_{i1}$ and $p_{i2}$ indicate the split powers for the  stream 1 and stream 2 of user $i$, respectively, and they must satisfy $p_{i1}+p_{i2}=1$ for $i\in\{1, 2\}$. Finally, utilizing superposition coding, these two signals are linearly combined and transmitted with the power $P_i^1$. Note that $P_{i}^j$ (for $i\in\{1, 2\}$ and $j\in\{1, 2\}$) shows the allocated power for user $i$ at mini-slot $j$. The received signal at the BS in the first mini-slot equals $y_{BS}^1=\sqrt{h_1P_1^1}(\sqrt{p_{11}}s_{11}+\sqrt{p_{12}}s_{12})+\sqrt{h_2P_2^1}(\sqrt{p_{21}}s_{21}+\sqrt{p_{22}}s_{22})+z_{BS}$. At the same time, user $i$ receives the transmitted signal of the other user $i'$. \par
At the second mini-slot, after applying self interference cancellation at each user, the received signal from the other user is given by $y_i=\sqrt{h_3P_{i'}^1}(\sqrt{p_{i'1}}s_{i'1}+\sqrt{p_{i'2}}s_{i'2})+z_{i}$ for $i\in\{1, 2\}$. Then each user normalizes and amplifies the received signal with the gain $\rho_i=\frac{P_i^2}{h_3P_{i'}^1+\sigma_i^2}$. Note that we must have $P_i^1+P_i^2=\bar{P}_{i}$ for $i\in\{1, 2\}$. Finally, users forward the amplified signals $\sqrt{\rho_i}y_i$ to the BS in the same frequency resources.  Assuming a strong link between two users, the received signal at the BS in the second  mini-slot is given by $y_{BS}^2=\sqrt{h_1\rho_1}y_{1}+\sqrt{h_2\rho_2}y_{2}+z_{BS}=\sqrt{h_1P_1^2}(\sqrt{p_{21}}s_{21}+\sqrt{p_{22}}s_{22})+\sqrt{h_2P_2^2}(\sqrt{p_{11}}s_{11}+\sqrt{p_{12}}s_{12})+z_{BS}$. \par
Now, the proposed detection scheme in Fig. \ref{detection-scheme-C-RSMA} is applied in order to decode two symbols of each user according to the received signals at two mini-slots. We decode symbols of users alternately with the following order $s_{11},s_{21},s_{12}$, and $s_{22}$.  At the proposed detection scheme, firstly, the MRC is applied for the received signals at two mini-slot in order to detect $s_{11}$ with the achievable rate $R_{11}=\frac{1}{2}\log_2(1+\frac{h_1P_1^1p_{11}}{h_1P_1^1p_{12}+h_2P_2^1p_{21}+h_2P_2^1p_{22}+\sigma_{BS}^2}+\frac{h_2P_2^2p_{11}}{h_2P_2^2p_{12}+h_1P_1^2p_{21}+h_1P_1^2p_{22}+\sigma_{BS}^2})$. Then, the interference of $s_{11}$ is subtracted from the received signals by applying SIC, and we again perform MRC by combining these two new received signals in order to detect $s_{21}$. This lead to  achievable rate of $R_{21}=\frac{1}{2}\log_2(1+\frac{h_2P_2^1p_{21}}{h_2P_2^1p_{22}+h_1P_1^1p_{12}+\sigma_{BS}^2}+\frac{h_1P_1^2p_{21}}{h_1P_1^2p_{22}+h_2P_2^2p_{12}+\sigma_{BS}^2})$ for $s_{21}$. Next, we apply SIC to remove the interference of $s_{21}$. Now, by applying MRC, we can decode $s_{12}$  with the achievable rate $R_{12}=\frac{1}{2}\log_2(1+\frac{h_1P_1^1p_{12}}{h_2P_2^1p_{22}+\sigma_{BS}^2}+\frac{h_2P_2^2p_{12}}{h_1P_1^2p_{22}+\sigma_{BS}^2})$. Finally, by applying SIC to remove $s_{12}$, we can perform MRC for the new received signals from two  mini-slots in order to decode $s_{22}$ with the achievable rate $R_{22}=\frac{1}{2}\log_2(1+\frac{h_2P_2^1p_{22}}{\sigma_{BS}^2}+\frac{h_1P_1^2p_{22}}{\sigma_{BS}^2})$. Now, we can write $R_1^{\mathrm{C-RSMA}}=R_{11}+R_{12}$, and $R_2^{\mathrm{C-RSMA}}=R_{21}+R_{22}$. Note that the original message of each user is reconstructed at the BS by combining the two decoded streams for that user. We define the normalized powers at each user $i$ for the first and second mini-slots as $p_i^1=\frac{P_i^1}{\bar{P}_{i}}$ and $p_i^2=\frac{P_i^2}{\bar{P}_{i}}$, respectively.
Also, by defining the SNR for each user as $\gamma_i=\frac{h_i\bar{P}_{i}}{\sigma^2}$ for $i\in\{1, 2\}$, the achievable rates for each user can be derived as (\ref{r1_crsma}) and (\ref{r2_crsma}), and the proof is completed.
\end{proof}

\begin{figure}[!t]
  \centering
  \includegraphics[width=0.8\textwidth]{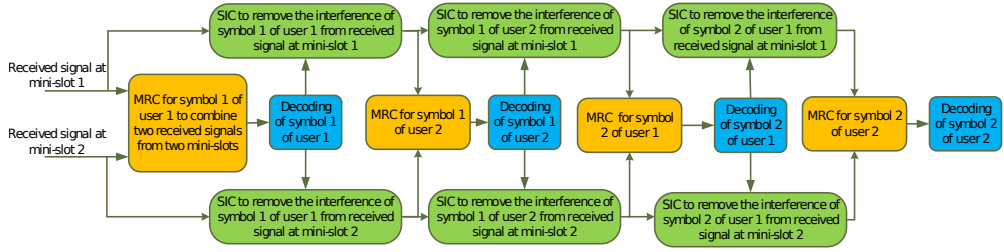}
  \caption{Proposed detection scheme for the C-RSMA scheme in order to decode the two symbols of each user according to the received signals in two mini-slots at the BS.}\label{detection-scheme-C-RSMA}
\end{figure}

The transmission complexity of the C-NOMA and C-RSMA schemes is more than C-OMA and non-cooperative schemes due to working in full-duplex mode in the first mini-slot. The transmission scheme of the C-RSMA is even more complex than C-NOMA as it requires splitting the messages of each user into two streams. Note that C-NOMA can be considered as a special case of C-RSMA in which there is no rate splitting at users. The detection complexity of the C-NOMA and C-RSMA schemes is more than C-OMA and non-cooperative schemes. In order to compare the complexity of the proposed detection methods for the C-NOMA and C-RSMA schemes, note that the detector of the C-NOMA scheme consists of two MRC and two SIC blocks, and the detector of the C-RSMA scheme consists of four MRC and six SIC blocks.

\section{Problem formulation for power allocation}
In this section, we formulate and solve two optimization problems in order to find optimal values for power allocation in C-NOMA, and C-RSMA schemes by maximizing the minimum rate of each user considering the proportional fairness coefficient. In this paper, we assume that perfect channel state information (CSI) is available at the BS. Therefore, based on this assumption and the formulation of max-min optimization problems for the uplink communication, we assume that the SIC can be performed perfectly.
\subsection{Power allocation for C-NOMA scheme}
In this subsection, we aim to find optimum power allocation $\bold{P}_i=[p_i^1,p_i^2]$  for $i\in\{1, 2\}$ that maximizes the minimum rate of each user in (\ref{r1_cnoma}) and (\ref{r2_cnoma}) with considering proportional fairness coefficient. 
 The optimization problem can be formulated as
\small
\begin{alignat}{2}
\text{(P1):~~~~    }&\underset{\bold{P}_1,\bold{P}_2}{\text{max}}~~~        && \min (R_1^{\mathrm{C-NOMA}},fR_2^{\mathrm{C-NOMA}})\label{eq:obj_fun}\\
&\text{s.t.} &      &   p_i^1+p_i^2\leq2, ~~\forall i\in\{1, 2\}, \label{eq:constraint-sum}\\
               &&& p_i^1>0,~ p_i^2>0,  \label{eq:constraint+}
\end{alignat}
\normalsize
where $f$ indicates the proportional fairness coefficient and regarding different rate requirements of users, we can consider different values for this parameter. Note that when $f=1$, this optimization problem reduces to the well-known mix-min problem. Constraint (\ref{eq:constraint-sum}) shows the total power consumption constraint in each user at two mini-slots. \par
The term $R_1^{\mathrm{C-NOMA}}$ in objective function at (\ref{eq:obj_fun}) is not concave with respect to power variables. Hence, (P1) is a non-convex optimization problem and can not be solved using standard convex optimization methods. In order to solve this problem, we propose the SCA-GP method in which we transform our optimization problem into a GP problem at each iteration of the SCA method. It has been proved \cite{boyd2004convex} that the GP problem can be transformed into a convex problem. In the SCA method, an approximation of the non-concave terms $R_1^{\mathrm{C-NOMA}}$ and $R_2^{\mathrm{C-NOMA}}$ is maximized iteratively.  Note that these approximations must be non-decreasing in iterations in order to guarantee the convergence of the SCA method. In the following lemma, we introduce a convex function that helps us to find an appropriate approximation.

\begin{lemma}\label{lemma1}
Consider $x$ and $y$ as two variables, and $a$, and $b$ as constants. The function  $f(x,y)=\ln(1+\frac{a}{x}+\frac{b}{y})$,
is a convex function with respect to $x$ and $y$.
\end{lemma}
\begin{proof}
See Appendix \ref{lemma 1}. 
\end{proof}
  In the following proposition, we utilize Lemma \ref{lemma1} in order to find an  approximation for $R_1^{\mathrm{C-NOMA}}$ and $R_2^{\mathrm{C-NOMA}}$. Note that in the $l^{\text{th}}$ iteration of SCA method, the power allocation coefficients of each user are indicated by $p_i^{1,l},~ p_i^{2,l}$ for  $i\in\{1, 2\}$. Also, the lower-bound (LB) approximated achievable rate for $R_1^{\mathrm{C-NOMA}}$ and $R_2^{\mathrm{C-NOMA}}$ at the iteration $l$ are shown by $R_{\text{lb},1}^{l,\mathrm{C-NOMA}}$ and $R_{\text{lb},2}^{l,\mathrm{C-NOMA}}$, respectively.
  
\begin{proposition}\label{prop-cnoma-opt}
For any given power allocation at the $l^{\text{th}}$ iteration, i.e., $[p_i^{1,l},p_i^{2,l}]$ for  $i\in\{1, 2\}$, one approximated non-decreasing lower-bound of the achievable rate of user 1 and user 2 in C-NOMA scheme and at the $(l+1)^{\text{th}}$ iteration of the SCA method is given by
\small
\begin{equation}
\begin{split}\label{rate-cnoma_lb}
R_{\text{lb},i}^{l+1,\mathrm{C-NOMA}}=\frac{1}{2}\log_2(1+\frac{1}{x_i^l}+\frac{1}{y_i^l})+\frac{1}{2}d_{i1}^l(x_i^{l+1}-x_i^l)+\frac{1}{2}d_{i2}^l(y_i^{l+1}-y_i^l),~i\in\{1, 2\},
\end{split}
\end{equation}
\normalsize
where $d_{i1}^l$ and $d_{i2}^l$  are partial derivatives as $d_{i1}^l=\frac{-1}{\ln{2}\big((x_i^l)^2+x_i^l+(\frac{x_i^l}{y_i^l})\big)}$ and $d_{i2}^l=\frac{-1}{\ln{2}\big((y_i^l)^2+y_i^l+(\frac{y_i^l}{x_i^l})\big)}$, and $x_1^l=\frac{\gamma_2p_2^{1,l}}{\gamma_1p_1^{1,l}}+\frac{1}{\gamma_1p_1^{1,l}}$, $y_1^l=\frac{\gamma_1p_1^{2,l}}{\gamma_2p_2^{2,l}}+\frac{1}{\gamma_2p_2^{2,l}}$, $x_2^l=\frac{1}{\gamma_1p_1^{2,l}}$, and $y_2^l=\frac{1}{\gamma_2p_2^{1,l}}$.
\end{proposition}
\begin{proof}
See Appendix \ref{appendix-cnoma}.  
\end{proof}
This proposition shows that the achievable rates of user 1  ($R_1^{\mathrm{C-NOMA}}$) and user 2 ($R_2^{\mathrm{C-NOMA}}$) in (\ref{r1_cnoma}) and (\ref{r2_cnoma}) are lower-bounded by  $R_{\text{lb},1}^{l+1,\mathrm{C-NOMA}}$ and $R_{\text{lb},2}^{l+1,\mathrm{C-NOMA}}$ in (\ref{rate-cnoma_lb}), respectively. It then follows that the optimal value of (P1) is lower-bounded by the optimal value of the problem 
\small
\begin{alignat}{2}
\text{(P1.1):~~~~}&\underset{\bold{P}_1,\bold{P}_2}{\text{max}}~~~        && \min (R_{\text{lb},1}^{l+1,\mathrm{C-NOMA}},fR_{\text{lb},2}^{l+1,\mathrm{C-NOMA}})\label{eq:obj_fun-p1.1}\\
&\text{s.t.}&     &  (\ref{eq:constraint-sum}), (\ref{eq:constraint+}).\nonumber
\end{alignat}
\normalsize
Problem (P1.1) is not still a convex optimization problem due to the non-concavity of $R_{\text{lb},1}^{l+1,\mathrm{C-NOMA}}$ with respect to power allocation coefficients. In the continue, we transform  (P1.1) into a convex problem by applying the variable change method. First of all, we rewrite the problem (P1.1) as
\small
\begin{alignat}{2}
\text{(P1.2):~~~~}&\underset{\bold{P}_1,\bold{P}_2,\eta}{\text{max}}~~~        && \eta \label{eq:obj_fun-p1.2}\\
 &\text{s.t.}&& R_{\text{lb},1}^{l+1,\mathrm{C-NOMA}}\geq\eta, \label{eq:constraint-R1-cnoma}\\
 &&& fR_{\text{lb},2}^{l+1,\mathrm{C-NOMA}}\geq\eta, \label{eq:constraint-R2-cnoma}\\
&&     &  (\ref{eq:constraint-sum}), (\ref{eq:constraint+}).\nonumber
\end{alignat}
\normalsize
It can be easily proved that (P1.2) is equivalent to (P1.1). Indeed, at the optimal solution for (P1.2), we must have $\eta=\min (R_{\text{lb},1}^{l+1,\mathrm{C-NOMA}},fR_{\text{lb},2}^{l+1,\mathrm{C-NOMA}})=R_{\text{lb},1}^{l+1,\mathrm{C-NOMA}}=fR_{\text{lb},2}^{l+1,\mathrm{C-NOMA}}$, and hence we can continue with (P1.2). In order to solve (P1.2), first we rewrite the constraints (\ref{eq:constraint-R1-cnoma}) and (\ref{eq:constraint-R2-cnoma}) based on the derived approximations in (\ref{rate-cnoma_lb}) for the achievable rates of users  as
\small
\begin{equation}\label{13}
    \frac{1}{2}\log_2(1+\frac{1}{x_1^l}+\frac{1}{y_1^l})+\frac{1}{2}d_{11}^l(\frac{\gamma_2p_2^{1,l+1}}{\gamma_1p_1^{1,l+1}}+\frac{1}{\gamma_1p_1^{1,l+1}}-x_1^l)+\frac{1}{2}d_{12}^l(\frac{\gamma_1p_1^{2,l+1}}{\gamma_2p_2^{2,l+1}}+\frac{1}{\gamma_2p_2^{2,l+1}}-y_1^l)\geq\eta,
\end{equation}
\begin{equation}\label{14}
    \frac{1}{2}\log_2(1+\frac{1}{x_2^l}+\frac{1}{y_2^l})+\frac{1}{2}d_{21}^l(\frac{1}{\gamma_1p_1^{2,l+1}}-x_2^l)+\frac{1}{2}d_{22}^l(\frac{1}{\gamma_2p_2^{1,l+1}}-y_2^l)\geq \frac{\eta}{f},
\end{equation}
\normalsize
respectively. Now, from (\ref{13}) and (\ref{14}), we can easily see that constraints (\ref{eq:constraint-R1-cnoma}) and (\ref{eq:constraint-R2-cnoma}) are in the form of a posynomial less than or equal to a monomial. As a consequence, according to \cite{boyd2004convex}, (P1.2) is a GP problem which can be transformed into a convex problem by performing the variable change as  $w_{i}^{1,l+1}=\ln{(p_i^{1,l+1})}$, and $w_{i}^{2,l+1}=\ln{(p_i^{2,l+1})}$ for $i\in\{1, 2\}$, and $\mu=\ln{(\eta)}$. After the variable change and some manipulations, by taking the logarithm of both sides of the inequalities in (\ref{13}) and (\ref{14}) we have
\small
\begin{equation}\label{15}
\begin{split}
    \log\Big(\exp(\mu)-\frac{1}{2}d_{11}^l(\frac{\gamma_2}{\gamma_1}&\exp(w_2^{1,l+1}-w_1^{1,l+1})
    +\frac{1}{\gamma_1}\exp(-w_1^{1,l+1}))-\frac{1}{2}d_{12}^l(\frac{\gamma_1}{\gamma_2}\exp(w_1^{2,l+1} -w_2^{2,l+1})\\&+\frac{1}{\gamma_2}\exp(-w_2^{2,l+1}))\Big)
    \leq 
    \log\Big(\frac{1}{2}\log_2(1+\frac{1}{x_1^l}+\frac{1}{y_1^l})+\frac{1}{2}d_{11}^l(-x_1^l)+\frac{1}{2}d_{12}^l(-y_1^l)\Big),
\end{split}
\end{equation}
\begin{equation}\label{16}
\begin{split}
    \log \Big( \frac{\exp(\mu)}{f}-\frac{1}{2}d_{21}^l(\frac{1}{\gamma_1}&\exp(-w_1^{2,l+1}))-\frac{1}{2}d_{22}^l(\frac{1}{\gamma_2}\exp(-w_2^{1,l+1}))\leq  \\& \log \Big(\frac{1}{2}\log_2(1+\frac{1}{x_2^l}+\frac{1}{y_2^l})+\frac{1}{2}d_{21}^l(-x_2^l)+\frac{1}{2}d_{22}^l(-y_2^l)\Big).
    \end{split}
\end{equation}
\normalsize
Also, the constraint (\ref{eq:constraint-sum}) will be
\small
\begin{equation}\label{17}
    \log\Big(\exp(w_i^{1,l+1})+\exp(w_i^{2,l+1})\Big)\leq \log(2), ~~\forall i\in\{1, 2\}.
\end{equation}
\normalsize
Since the logarithm of the summation of the exponential functions is a convex function \cite{boyd2004convex}, the left-hand side of the inequalities (\ref{15}), (\ref{16}), and (\ref{17}) are in the form of a convex function which means that all of the constraints of the (P1.2) are in the form of a convex function less than a constant. Hence, 
we have an optimization problem that is convex, and can be efficiently solved by standard convex optimization solvers. In this paper, we utilize CVX solver \cite{cvx} that applies the interior point method to solve these convex problems. Finally, the optimum values for power allocation and objective function of (P1.1) can be derived from $p_i^{1,l+1}=\exp{(w_{i}^{1,l+1})}$, and $p_i^{2,l+1}=\exp{(w_{i}^{2,l+1})}$ for $i\in\{1, 2\}$, and $\eta=\exp{(\mu)}$. 
 One can see the summary of the proposed iterative method for solving (P1) in Algorithm 1. \par
 Note that for deriving the approximations in (\ref{rate-cnoma_lb}) for achievable rates of users in Proposition \ref{prop-cnoma-opt}, we applied the first-order Taylor expansion for the approximation of a convex function, and therefore the achievable rates of users at each iteration $l$ of the SCA method are greater than the approximated rate, i.e., $R_{i}^{l+1,\mathrm{C-NOMA}}>R_{\text{lb},i}^{l+1,\mathrm{C-NOMA}}$ for $i\in\{1, 2\}$. Moreover, since we maximize  the approximated rates at each iteration $l$ of the Algorithm 1, we can write $R_{\text{lb},i}^{l+1,\mathrm{C-NOMA}}>R_{\text{lb},i}^{l,\mathrm{C-NOMA}}$. Considering this point that the Taylor expansion of a function around an initial point is exactly equal to the value of the original function at that point, i.e., $R_{\text{lb},i}^{l,\mathrm{C-NOMA}}=R_{i}^{l,\mathrm{C-NOMA}}$, we can conclude that $R_{i}^{l+1,\mathrm{C-NOMA}}>R_{i}^{l,\mathrm{C-NOMA}}$ which means that the achievable rates of users $R_{i}^{l,\mathrm{C-NOMA}}$ are non-decreasing with $l$ in the proposed method. Therefore, since the objective function is non-decreasing over iterations and is globally upper-bounded by the optimal value of the original problem (P1), the proposed sub-optimal algorithm is guaranteed to converge.
\begin{algorithm}
\caption{Iterative power allocation for C-NOMA scheme in (P1) utilizing the proposed SCA-GP method.}\label{alg1}
\begin{algorithmic}[1]
\State  Initialize the allocated powers $[p_i^{1,l},p_i^{2,l}]$ for  $i\in\{1, 2\}$, and let $l=0$.
\Repeat
\State Find the lower-bound approximations for achievable rates of user 1 and user 2 from (\ref{rate-cnoma_lb}) in Proposition \ref{prop-cnoma-opt}.
\State Transform (P1.1) into (P1.2) which is in the standard form of a GP problem. 
\State Perform the variable changes as  $w_{i}^{1,l+1}=\ln{(p_i^{1,l+1})}$, and $w_{i}^{2,l+1}=\ln{(p_i^{2,l+1})}$ for $i\in\{1, 2\}$, and $\mu=\ln{(\eta)}$, and take the logarithm of the both sides of the inequalities in (P1.2) to reach inequalities (\ref{15}), (\ref{16}), and (\ref{17}) that are convex constraints.
\State Solve this convex problem utilizing the interior-point method with the CVX solver \cite{cvx}, and find the optimal values as $p_i^{1,l+1}=\exp{(w_{i}^{1,l+1})}$, and $p_i^{2,l+1}=\exp{(w_{i}^{2,l+1})}$ for $i\in\{1, 2\}$, and $\eta=\exp{(\mu)}$.
\State Update $l=l+1$.
\Until {convergence or a maximum number of iterations is reached.}
\end{algorithmic}
\end{algorithm}

\subsection{Power allocation for C-RSMA scheme}
In this subsection, we aim to find optimum power allocation coefficients $\bold{P}_i=[p_i^1,p_i^2,p_{i1},p_{i2}]$ (for $i\in\{1, 2\}$)  at the C-RSMA scheme that maximize the minimum rate of each user in (\ref{r1_crsma}) and (\ref{r2_crsma}) with considering proportional fairness coefficient. The optimization problem can be formulated as
\small
\begin{alignat}{2}
\text{(P2):~~~~}&\underset{\bold{P}_1,\bold{P}_2}{\text{max}}~~~        && \min (R_1^{\mathrm{C-RSMA}},fR_2^{\mathrm{C-RSMA}})\label{eq:obj_fun_RSMA-P2}\\
&\text{s.t.} &      &   p_i^1+p_i^2\leq2, ~~\forall i\in\{1, 2\}, \label{eq:constraint-sum-rsma}\\
&&& p_{i1}+p_{i2}\leq1, ~~\forall i\in\{1, 2\}, \label{eq:constraint-sum-rsma-stream}\\
               &&& p_i^1>0,~ p_i^2>0, ~ p_{i1}>0, ~ p_{i2}>0, ~~\forall i\in\{1, 2\}, \label{eq:constraint+rsma}
\end{alignat}
\normalsize
where $f$ indicates the proportional fairness coefficient. Constraint (\ref{eq:constraint-sum-rsma}) shows the total power consumption constraint in each user at the two mini-slots. Also, constraint (\ref{eq:constraint-sum-rsma-stream}) indicates the power allocation constraint between two symbols at each user. \par
One can see that $R_1^{\mathrm{C-RSMA}}$ and $R_2^{\mathrm{C-RSMA}}$ in (\ref{r1_crsma}) and (\ref{r2_crsma}) are not concave functions with respect to power allocation variables. Hence, (P2) is a non-convex optimization problem. In order to solve this problem, similar to the solution for (P1), we propose the SCA-GP method in which (P2) is transformed into a GP problem at each iteration of the SCA method.
In SCA method, an approximation of the non-concave terms $R_1^{\mathrm{C-RSMA}}$ and $R_2^{\mathrm{C-RSMA}}$ is maximized iteratively.
  In the following proposition, we utilize Lemma \ref{lemma1} in order to find this  approximation. Note that in the $l^{\text{th}}$ iteration of SCA method, the power allocation coefficients of each user are indicated by $p_i^{1,l},~ p_i^{2,l},~p_{i1}^l,~p_{i2}^l$ for  $i\in\{1, 2\}$. Also, LB approximated achievable rates for $R_1^{\mathrm{C-RSMA}}$ and $R_2^{\mathrm{C-RSMA}}$ at the iteration $l$ are shown by $R_{\text{lb},1}^{l,\mathrm{C-RSMA}}$ and $R_{\text{lb},2}^{l,\mathrm{C-RSMA}}$, respectively.
  
\begin{proposition}\label{prop-crsma-opt}
For any given power allocation at the $l^{\text{th}}$ iteration, i.e., $[p_i^{1,l},p_i^{2,l},p_{i1}^l,p_{i2}^l]$ for  $i\in\{1, 2\}$, one approximated non-decreasing lower-bound of the achievable rate of user 1 and user 2 in C-RSMA scheme and at the $(l+1)^{\text{th}}$ iteration of the SCA method is given by
\small
\begin{equation}
\begin{split}\label{rate-crsma_lb}
&R_{\text{lb},i}^{l+1,\mathrm{C-RSMA}}=\frac{1}{2}\log_2(1+\frac{1}{x_i^{1,l}}+\frac{1}{y_i^{1,l}})+\frac{1}{2}\log_2(1+\frac{1}{x_i^{2,l}}+\frac{1}{y_i^{2,l}})\\&+\frac{1}{2}d_{i1}^{1,l}(x_i^{1,l+1}-x_i^{1,l})+\frac{1}{2}d_{i2}^{1,l}(y_i^{1,l+1}-y_i^{1,l})+\frac{1}{2}d_{i1}^{2,l}(x_i^{2,l+1}-x_i^{2,l})+\frac{1}{2}d_{i2}^{2,l}(y_i^{2,l+1}-y_i^{2,l}), ~i\in\{1, 2\},
\end{split}
\end{equation}
\normalsize
where $d_{i1}^{s,l}$ and $d_{i2}^{s,l}$ are partial derivatives as $d_{i1}^{s,l}=\frac{-1}{\ln{2}\big((x_i^{s,l})^2+x_i^{s,l}+(\frac{x_i^{s,l}}{y_i^{s,l}})\big)}$ and $d_{i2}^{s,l}=\frac{-1}{\ln{2}\big((y_i^{s,l})^2+y_i^{s,l}+(\frac{y_i^{s,l}}{x_i^{s,l}})\big)}$, for $i\in\{1, 2\}$ and $s\in\{1, 2\}$, and $x_1^{1,l}=\frac{p_{12}^{l}}{p_{11}^{l}}+\frac{p_2^{1,l}p_{21}^l\gamma_2}{p_1^{1,l}p_{11}^l\gamma_1}+\frac{p_2^{1,l}p_{22}^l\gamma_2}{p_1^{1,l}p_{11}^l\gamma_1}+\frac{1}{p_1^{1,l}p_{11}^l\gamma_1}$, $y_1^{1,l}=\frac{p_{12}^{l}}{p_{11}^{l}}+\frac{p_1^{2,l}p_{21}^{l}\gamma_1}{p_2^{2,l}p_{11}^{l}\gamma_2}+\frac{p_1^{2,l}p_{22}^{l}\gamma_1}{p_2^{2,l}p_{11}^{l}\gamma_2}+\frac{1}{p_2^{2,l}p_{11}^{l}\gamma_2}$, $x_1^{2,l}=\frac{p_2^{1,l}p_{22}^{l}\gamma_2}{p_1^{1,l}p_{12}^{l}\gamma_1}+\frac{1}{p_1^{1,l}p_{12}^{l}\gamma_1}$, $y_1^{2,l}=\frac{p_1^{2,l}p_{22}^{l}\gamma_1}{p_2^{2,l}p_{12}^{l}\gamma_2}+\frac{1}{p_2^{2,l}p_{12}^{l}\gamma_2}$, $x_2^{1,l}=\frac{p_{22}^{l}}{p_{21}^{l}}+\frac{p_1^{1,l}p_{12}^{l}\gamma_1}{p_2^{1,l}p_{21}^{l}\gamma_2}+\frac{1}{p_2^{1,l}p_{21}^{l}\gamma_2}$, $y_2^{1,l}=\frac{p_{22}^{l}}{p_{21}^{l}}+\frac{p_2^{2,l}p_{12}^{l}\gamma_2}{p_1^{2,l}p_{21}^{l}\gamma_1}+\frac{1}{p_1^{2,l}p_{21}^{l}\gamma_1}$, $x_2^{2,l}=\frac{1}{p_2^{1,l}p_{22}^{l}\gamma_2}$, and $y_2^{2,l}=\frac{1}{p_1^{2,l}p_{22}^{l}\gamma_1}$.
\end{proposition}
\begin{proof}
The proof for this proposition is similar to the proof of Proposition \ref{prop-cnoma-opt}, and hence we remove it.  
\end{proof}
This proposition shows that the rates of user 1  ($R_1^{\mathrm{C-RSMA}}$) in (\ref{r1_crsma}) and user 2 ($R_2^{\mathrm{C-NOMA}}$) in (\ref{r2_crsma}) are lower-bounded  by the $R_{\text{lb},1}^{l+1,\mathrm{C-RSMA}}$ and $R_{\text{lb},2}^{l+1,\mathrm{C-RSMA}}$ in (\ref{rate-crsma_lb}), respectively. It then follows that the optimal value of (P2) is lower-bounded by the optimal value of the following problem 
\small
\begin{alignat}{2}
\text{(P2.1):~~~~}&\underset{\bold{P}_1,\bold{P}_2}{\text{max}}~~~        && \min (R_{\text{lb},1}^{l+1,\mathrm{C-RSMA}},fR_{\text{lb},2}^{l+1,\mathrm{C-RSMA}})\label{eq:obj_fun-p1.1}\\
&\text{s.t.}&     &  (\ref{eq:constraint-sum-rsma}), (\ref{eq:constraint-sum-rsma-stream}), (\ref{eq:constraint+rsma}).\nonumber
\end{alignat}
\normalsize
Problem (P1.2) is still a non-convex problem due to the non-concavity of its objective function. In order to solve (P2.1), first of all, we rewrite problem (P2.1) as
\small
\begin{alignat}{2}
\text{(P2.2):~~~~}&\underset{\bold{P}_1,\bold{P}_2,\eta}{\text{max}}~~~        && \eta \label{eq:obj_fun-p2.2}\\
 &\text{s.t.}&& R_{\text{lb},1}^{l+1,\mathrm{C-RSMA}}\geq\eta, \label{eq:constraint-R1-crsma}\\
 &&& fR_{\text{lb},2}^{l+1,\mathrm{C-RSMA}}\geq\eta, \label{eq:constraint-R2-crsma}\\
&&     &  (\ref{eq:constraint-sum-rsma}), (\ref{eq:constraint-sum-rsma-stream}), (\ref{eq:constraint+rsma}).\nonumber
\end{alignat}
\normalsize
Note that constraints (\ref{eq:constraint-R1-crsma}) and (\ref{eq:constraint-R2-crsma}) are in the form of a posynomial less than or equal to a monomial. Hence, (P2.2) is a standard GP problem that can be transformed into a convex problem.
One can see the summary of the proposed iterative method for solving (P2) in Algorithm 2. Similar to Algorithm 1, 
we can prove that the objective function of (P2.2) is non-decreasing over iterations and is globally upper-bounded by the optimal value of (P2). Therefore, the proposed sub-optimal algorithm is guaranteed to converge.  

\subsection{Computational Complexity Analysis}
Now, the computational complexity of the proposed iterative power allocations for (P1) in Algorithm 1, and (P2) in Algorithm 2 are presented. At each iteration $l$ of the proposed SCA-GP methods in Algorithms 1 and 2, computational complexity is dominated by solving convex problems in (P1.2) and (P2.2), respectively, which are solved by the interior point method. According to \cite{boyd2004convex}, the interior point method
method requires $\log(\frac{n_c}{t^0\varrho})/\log\varepsilon$ number of iterations
(Newton steps) to solve a convex problem, where $n_c$ is the total
number of constraints, $t^0$ is the initial point for approximating the accuracy of the interior-point
method, $0<\varrho\ll1$ is
the stopping criterion, and $\varepsilon$ is used for updating the accuracy
of the interior point method. For problems (P1.2) and (P2.2), the number of constraints are equal to $n_{c,1}=4I$ and $n_{c,2}=7I$, respectively, in which $I$ indicates the number of users. Hence, the computational complexity of Algorithm 1 and 2 will be $\mathcal{O}(n_{\mathrm{iter,1}}(\frac{\log(\frac{I}{t^{0}\varrho})}{\log\varepsilon}))$ and $\mathcal{O}(n_{\mathrm{iter,2}}(\frac{\log(\frac{I}{t^{0}\varrho})}{\log\varepsilon}))$, in which $n_{\mathrm{iter,1}}$ and $n_{\mathrm{iter,2}}$ indicate the number of iterations for convergence of Algorithm 1 and 2, respectively.

\begin{algorithm}
\caption{Iterative power allocation for C-RSMA scheme in (P2) utilizing the proposed SCA-GP method.}\label{alg1}
\begin{algorithmic}[1]
\State  Initialize the allocated powers $[p_i^{1,l},p_i^{2,l},p_{i1}^l,p_{i2}^l]$ for  $i\in\{1, 2\}$, and let $l=0$.
\Repeat
\State Find the lower-bound approximations for achievable rates of user 1 and user 2 from (\ref{rate-crsma_lb}) in Proposition \ref{prop-crsma-opt}.
\State Transform (P2.1) into (P2.2) which is in the standard form of a GP problem. 
\State Perform the variable changes as  $w_{i}^{1,l+1}=\ln{(p_i^{1,l+1})}$,  $w_{i}^{2,l+1}=\ln{(p_i^{2,l+1})}$, $w_{i1}^{l+1}=\ln{(p_{i1}^{l+1})}$, and $w_{i2}^{l+1}=\ln{(p_{i2}^{l+1})}$ for $i\in\{1, 2\}$, and $\mu=\ln{(\eta)}$, and take logarithm of the objective function and both sides of inequalities in (P2.2).
\State Solve this convex problem utilizing the interior-point method with the CVX solver \cite{cvx}, and find the optimal values as $p_i^{1,l+1}=\exp{(w_{i}^{1,l+1})}$, $p_i^{2,l+1}=\exp{(w_{i}^{2,l+1})}$, $p_{i1}^{l+1}=\exp{(w_{i1}^{l+1})}$, and $p_{i2}^{l+1}=\exp{(w_{i2}^{l+1})}$ for $i\in\{1, 2\}$, and $\eta=\exp{(\mu)}$.
\State Update $l=l+1$.
\Until {convergence or a maximum number of iterations is reached.}
\end{algorithmic}
\end{algorithm}

\section{Outage Probability and Diversity Analysis}
 In this section, we derive the asymptotic outage probability and diversity order of the proposed C-NOMA and C-RSMA schemes. Note that the C-NOMA and C-RSMA schemes are proposed to improve both the achievable rate and diversity. In order to maximize the minimum rate of each user, we formulated two optimization problems for C-NOMA and C-RSMA schemes in the previous section, and found the optimum values of power allocation coefficients based on the realizations of the channels. When the CSI of channels is not available, the proposed power allocation algorithms in section IV are performed based on the large-scale fading of the links. In this case, a deep fading can decrease the achievable rate of a user below its required rate, and cause an outage in the system. The outage probability of user $i$ (for $i\in\{1, 2\}$) and transmission scheme $s$ (for $s\in\{\mathrm{OMA}, \mathrm{C-OMA}, \mathrm{NOMA}, \mathrm{C-NOMA}, \mathrm{RSMA}, \mathrm{C-RSMA}\}$) is given by 
 \small
\begin{equation}
    P_{i,out}^{s}= \mathrm{Pr}(R_{i}^{s}<R_{i}^{th}),
\end{equation}
\normalsize
where $R_{i}^{th}$ is the required threshold rate for user $i$ so that $R_{1}^{th}=f R_{2}^{th}$. Then, the outage probability of the system for transmission scheme $s$ is defined as
\small
\begin{equation}
     P_{out}^{s}=1- \mathrm{Pr}(R_{1}^{s}>R_{1}^{th}, R_{2}^{s}>R_{2}^{th}).
\end{equation}
\normalsize
In the continue, we aim to derive the high SNR approximations for the outage probability of the C-NOMA and C-RSMA schemes. For this end, firstly, we introduce the following lemma.

\begin{lemma}\label{lemma2}
Assume that X and Y are two exponentially
distributed random variables with the mean values of $\lambda_x$ and $\lambda_y$, respectively.
The cumulative density function (CDF) of the random variable
Z=X+Y is given by
\small
\begin{equation}
    F_Z(z)=1-\exp{(-\frac{z}{\lambda_y})}-\frac{\exp{(-\frac{z}{\lambda_x})}-\exp{(-\frac{z}{\lambda_y}})}{1-\frac{\lambda_y}{\lambda_x}}, ~z>0.
\end{equation}
\normalsize
\end{lemma}
\begin{proof}
The proof is straightforward, and hence we remove it. 
\end{proof}
\begin{proposition}\label{proposition-outage}
Asymptotic outage probability of the proposed C-NOMA and C-RSMA schemes are given by 
\small
\begin{equation}\label{out_cnoma}
  P_{out}^{\mathrm{C-NOMA},\infty}=\frac{0.5(2^{2R_{2}^{th}}-1)}{p_2^1p_1^2\Omega_1\Omega_2},
\end{equation}
and
\begin{equation}\label{out_crsma}
P_{out}^{\mathrm{C-RSMA},\infty}=\frac{0.5(2^{2R_{2}^{th}}-1)}{p_2^1p_1^2p_{22}^2\Omega_1\Omega_2},
\end{equation}
\normalsize
respectively, where $\Omega_i=\frac{\bar{P}_{i}\beta_0}{\sigma_{BS}^2d_i^\alpha
}$ for $i\in\{1, 2\}$ is the mean SNR value of each user.
\end{proposition}
\begin{proof}
See Appendix \ref{appendix-outage}.
\end{proof}
\begin{corollary}\label{corollary1}
Diversity orders of the proposed C-NOMA and C-RSMA schemes are two. 
\end{corollary}
\begin{proof}
Utilizing the asymptotic outage probability of the proposed schemes in (\ref{out_cnoma}) and (\ref{out_crsma}), and according to the definition of the diversity order (i.e., $D=\lim_{\mathrm{SNR}\rightarrow{\infty}}-\frac{\log P_{out}}{\log \mathrm{SNR}}$), it can be easily seen that the diversity orders of both schemes are two, and the proof is completed.
\end{proof}
Corollary \ref{corollary1} means that the proposed C-NOMA and C-RSMA schemes can achieve full diversity, and hence are more robust to channel variations compared with non-cooperative schemes.\\ 
\section{Numerical Results}
In this section, numerical results are provided in order to show the performance gain of the proposed cooperative schemes. The following default parameters are applied in the simulations except that we specify different values for them. We assume that the two users are close enough to establish a strong link with each other. Hence, we assume that user 1 and user 2 have distances $d_1=d_2=100~\mathrm{m}$ to the BS. The path loss exponent equals $\alpha=~3.7$. The communication bandwidth is $B=1~\mathrm{MHz}$  with a
carrier frequency at 5 GHz, and a noise power spectral density of $N_0=-174~\mathrm{dBm/Hz}$. 
 Also, we assume that the average transmit power at each user equals $\bar{P}_1=\bar{P}_2=20~\mathrm{dBm}$. The feasible initial values for  power allocation coefficients in C-NOMA scheme at Algorithm 1 are $[p_i^{1,0},p_i^{2,0}]=[1,1]$ for  $i\in\{1, 2\}$. Also, the feasible initial values for power allocation coefficients in C-RSMA scheme at Algorithm 2 are $[p_i^{1,0},p_i^{2,0},p_{i1}^0,p_{i2}^0]=[1,1,0.5,0.5]$ for  $i\in\{1, 2\}$. \par

Fig. \ref{fig:RR-d10} shows the rate region for the proposed cooperative schemes and their non-cooperative counterparts for the case that the difference between the channel power gains of two users is equal to $h_1-h_2=10 ~\mathrm{dB}$. We can also see line $R_1=fR_2$ for three different proportional fairness coefficients $f=1,\frac{1}{3},3$ in order to clearly show the performance difference among schemes when two users have different QoS requirements. In Fig. \ref{fig:rates-region-p8d10}, one can see the rate region for the case that transmit power of users is low, i.e., $\bar{P}_1=\bar{P}_2=8~\mathrm{dBm}$. It can be seen that cooperative schemes have significant performance gain rather than non-cooperative schemes for all three proportional fairness coefficients $f=1,\frac{1}{3},3$. 
Note that rate region curves for cooperative schemes are symmetric at this figure which is due to the fact that in cooperative schemes, user 1 with better channel conditions helps user 2 so that both users can reach a good achievable rate. 
Also, one can see that the rate region of C-NOMA (NOMA) is a subset of the rate region of C-RSMA (RSMA) and the rate region of C-OMA (OMA) is a subset of the rate region of C-NOMA (NOMA).
These superiorities are at the cost of increased complexity at the transmitter and receiver of C-NOMA (NOMA) and C-RSMA (RSMA) schemes. In Fig. \ref{fig:rates-region-p20d10}, one can see the rate region for the case that transmit power of users is high, i.e., $\bar{P}_1=\bar{P}_2=20~\mathrm{dBm}$. We can see that the superiorities of cooperative schemes rather than their non-cooperative counterparts have been decreased. The reason is that the increased transmit power leads to a better achievable rate for user 2 in non-cooperative schemes.  
 From Fig. \ref{fig:RR-d10}, we can also say that when two users have more different QoS requirements, the superiority of C-RSMA over C-NOMA increases. 
 
Note that the proposed transmission and detection schemes for C-NOMA and C-RSMA schemes are not optimal, and hence the derived achievable rates for them in Propositions 1 and 2 are not their capacity. The achievable
rate region is, by definition, smaller than or equal to the capacity region. This implies that any performance gains reported
in this paper are a lower bound on the gains that can be achieved via user cooperation. On the other hand, note that our proposed cooperative schemes achieve full diversity of two. According to the diversity-multiplexing tradeoff \cite{tradeoff}, increasing diversity order decreases the multiplexing gain. Indeed, due to working in two mini-slots, we have a loss
of spectral efficiency with coefficient $0.5$. As a consequence, at some points of the rate region, non-cooperative schemes have better rate performance rather than cooperative schemes.\par
\begin{figure}[t]
\begin{subfigure}{0.5\textwidth}
\includegraphics[height=5.3cm]{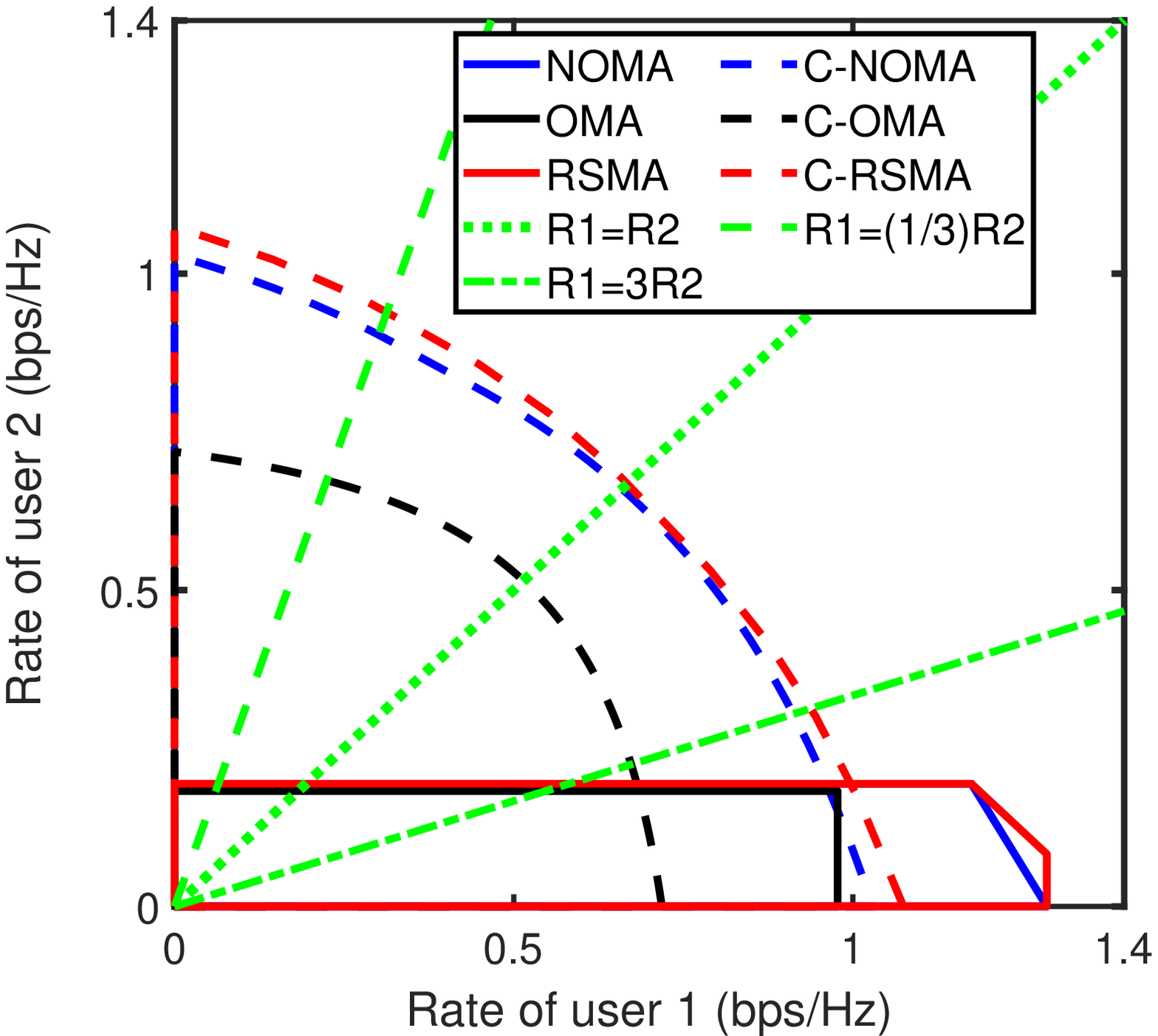} 
\caption{Low transmit power ($\bar{P}_1=\bar{P}_2=8~\mathrm{dBm}$)}
\label{fig:rates-region-p8d10}
\end{subfigure}
\begin{subfigure}{0.5\textwidth}
\includegraphics[height=5.3cm]{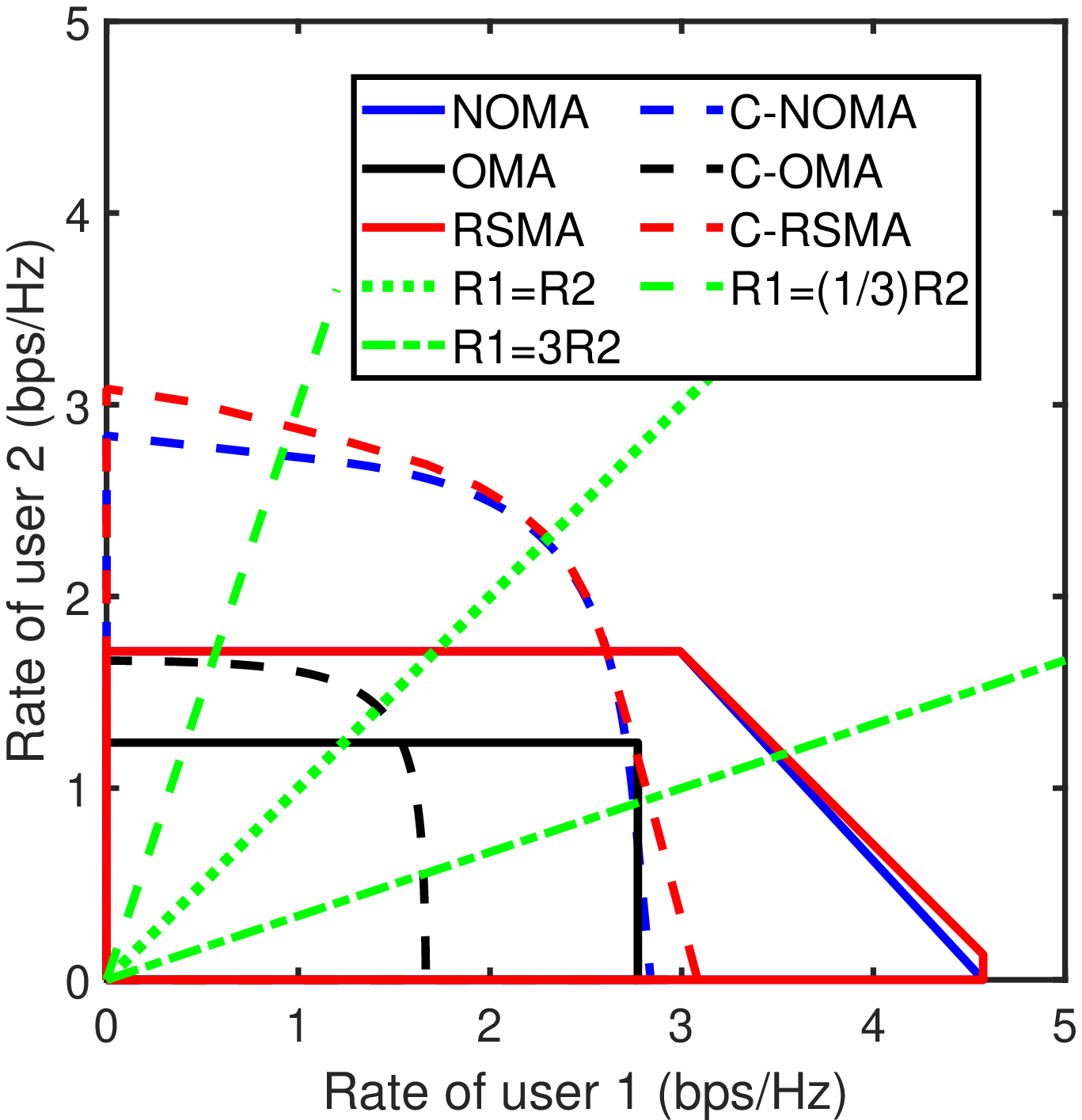}
\caption{High transmit power ($\bar{P}_1=\bar{P}_2=20~\mathrm{dBm}$)}
\label{fig:rates-region-p20d10}
\end{subfigure}
\caption{Rate region for the proposed cooperative schemes and their non-cooperative counterparts for the case that difference between the channel power gains of two users equals $h_1-h_2=10 ~\mathrm{dB}$.}
\label{fig:RR-d10}
\end{figure}

In Fig. \ref{fig:converg}, we can see the convergence of the proposed methods in Algorithm 1 and Algorithm 2. It is shown that the proposed algorithms converge in a few iterations which indicates the efficiency of the proposed SCA-GP method. Also, in this figure, one can see the performance gap between the proposed algorithms and the optimal values. In order to attain optimal values for objective functions of the C-NOMA and C-RSMA problems in (P1) and (P2), we performed an exhaustive search over all feasible power allocation variables. We can see in this figure that the proposed sub-optimal algorithms for (P1) and (P2) have small performance gaps with optimal values. In this figure, we assume that $\bar{P}_1=\bar{P}_2=20~\mathrm{dBm}$, $h_1-h_2=10 ~\mathrm{dB}$, and $f=1,\frac{1}{3}$. \par

\begin{figure}[t]
\begin{subfigure}{0.5\textwidth}
\includegraphics[width=1\linewidth, height=4.2cm]{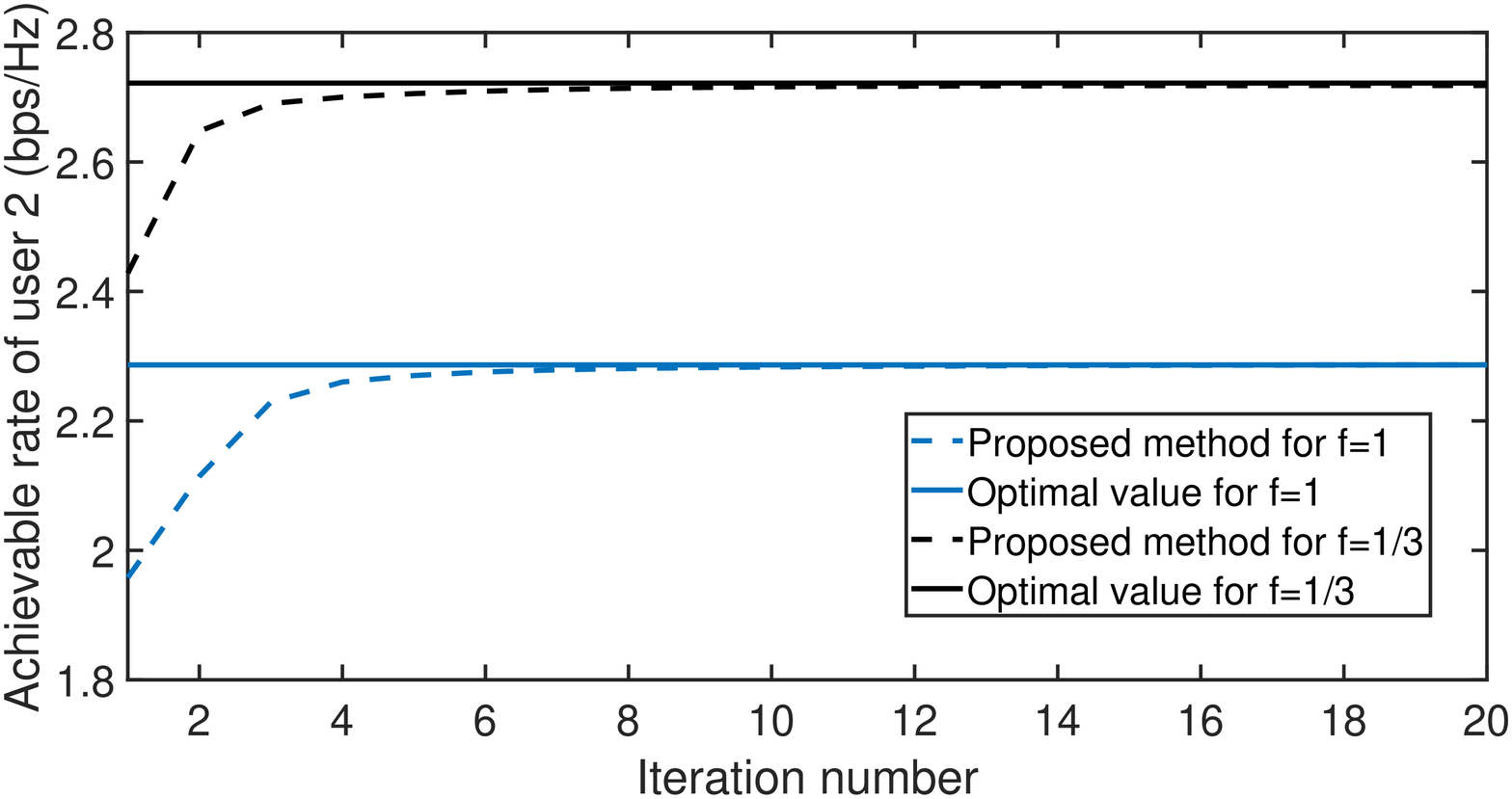}
\caption{Algorithm 1 for C-NOMA}
\label{fig:converg_CNOMA}
\end{subfigure}
\begin{subfigure}{0.5\textwidth}
\includegraphics[width=1\linewidth, height=4.2cm]{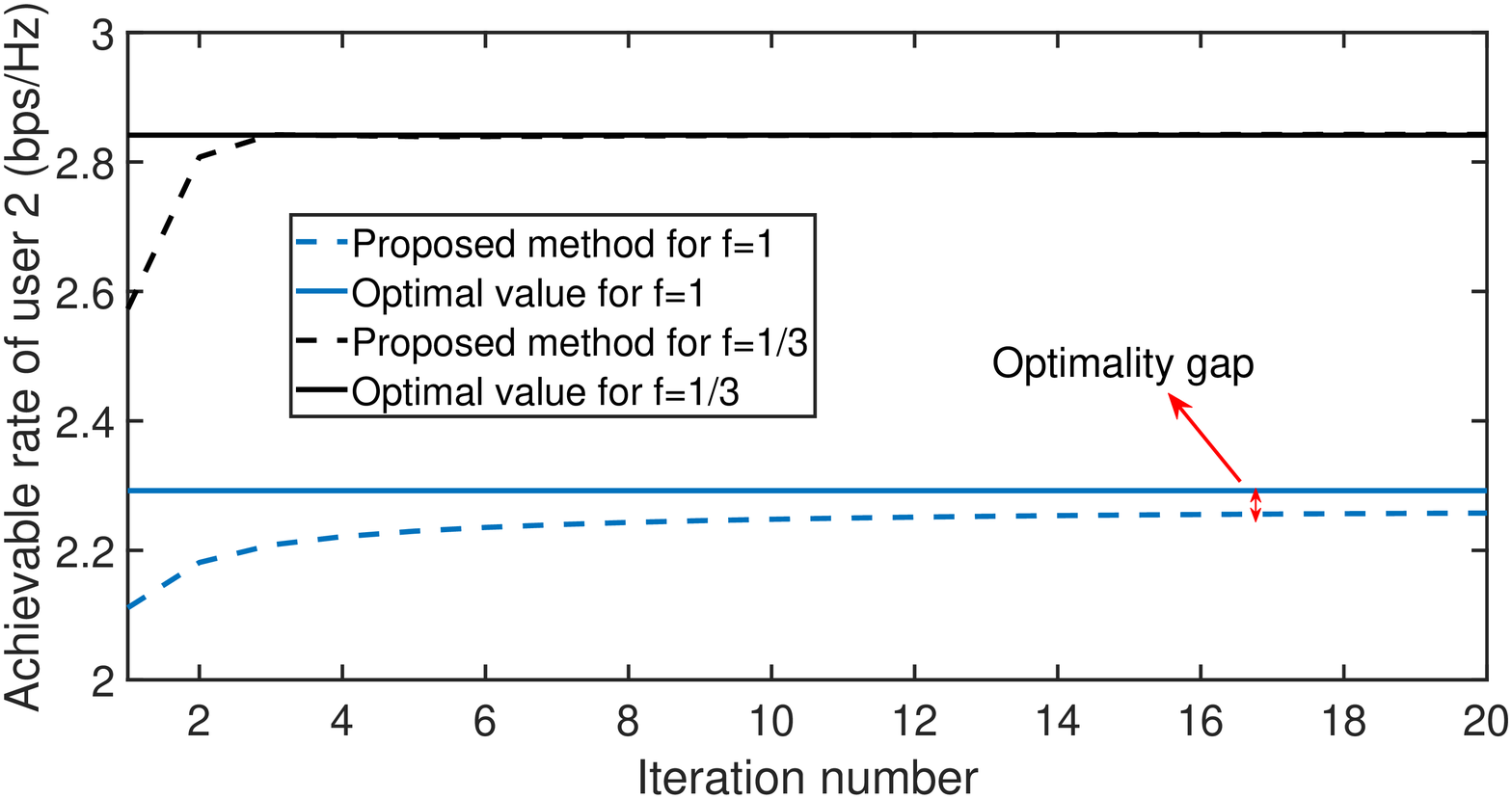}
\caption{Algorithm 2 for C-RSMA}
\label{fig:converg_CRSMA}
\end{subfigure}
\caption{Convergence of the proposed algorithms and their performance gap compared with optimal values.}
\label{fig:converg}
\end{figure}

Fig. \ref{fig:r2_d210_13} indicates the achievable rate of user 2 versus transmit power when the proportional fairness coefficient is $f=\frac{1}{3}$. We can see this achievable rate for C-NOMA and C-RSMA schemes with both the proposed algorithms and exhaustive search. Note that in our optimization problems, we maximize the minimum rate of users with the proportional fairness coefficient, and at the optimal solution of them, we have $R_1=fR_2$. Hence, in the following figures, we only show the achievable rate of user 2, and the achievable rate of user 1 can be easily derived as $R_1=fR_2$. Fig. \ref{fig:r2_d10_13} shows the case that $h_1-h_2=10 ~\mathrm{dB}$. It can be seen that the achievable rates of proposed algorithms for C-NOMA and C-RSMA schemes match with the rates with the exhaustive search which proves the effectiveness of the proposed algorithms for solving optimization problems. We can see that the superiority of C-NOMA and C-RSMA schemes compared with C-OMA increases with transmit power. Also, cooperative schemes have significant performance gains rather than non-cooperative schemes. Fig. \ref{fig:r2_d2_13} indicates the case that $h_1-h_2=2 ~\mathrm{dB}$, in which by increasing transmit power, non-cooperative schemes outperform cooperative ones. As a consequence, these figures indicate that cooperation is meaningful when  $h_1-h_2$ is big.\par

 \begin{figure}[t]
\begin{subfigure}{0.5\textwidth}
\includegraphics[width=1\linewidth, height=4.3cm]{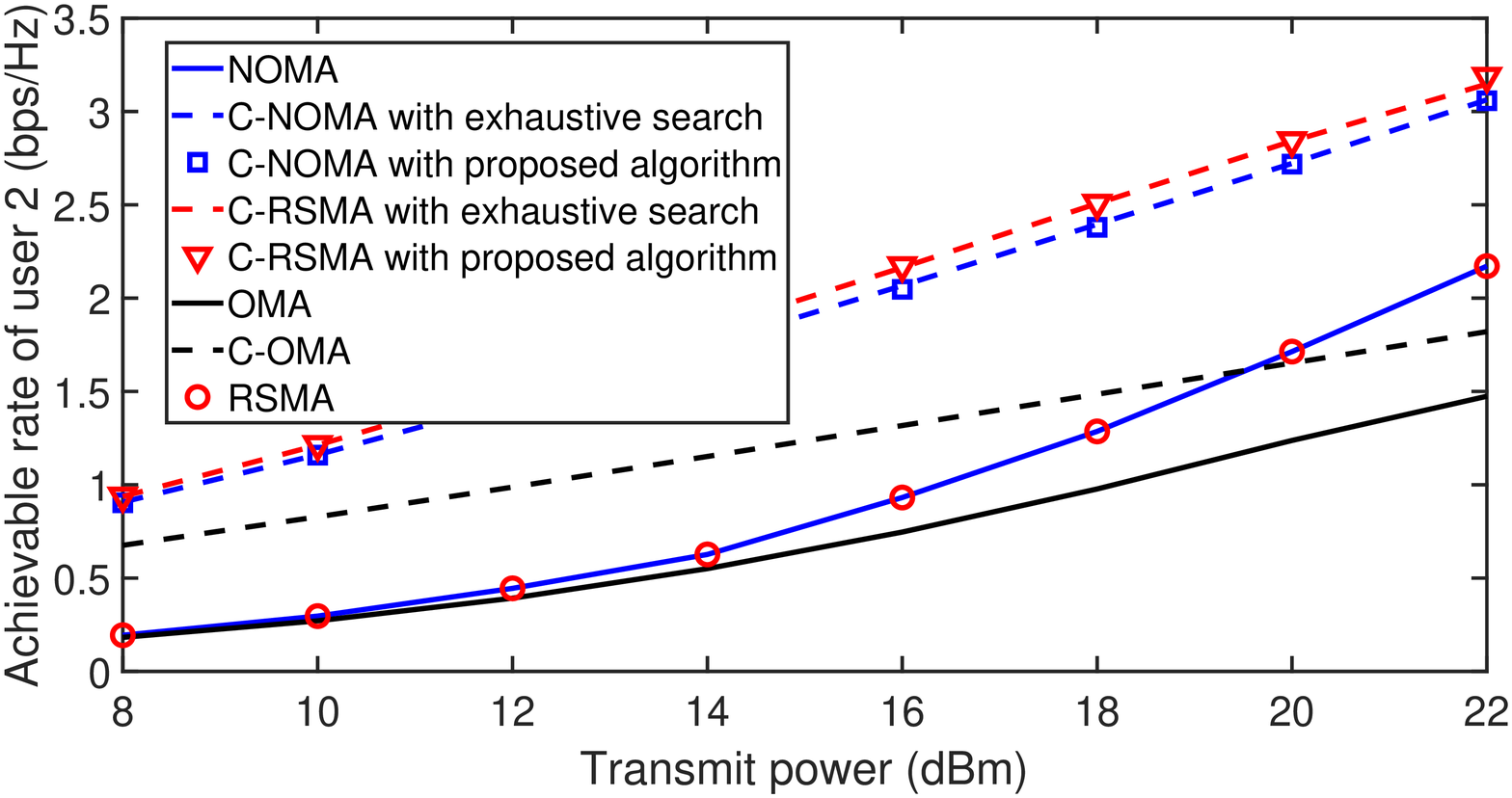} 
\caption{$h_1-h_2=10 ~\mathrm{dB}$}
\label{fig:r2_d10_13}
\end{subfigure}
\begin{subfigure}{0.5\textwidth}
\includegraphics[width=1\linewidth, height=4.3cm]{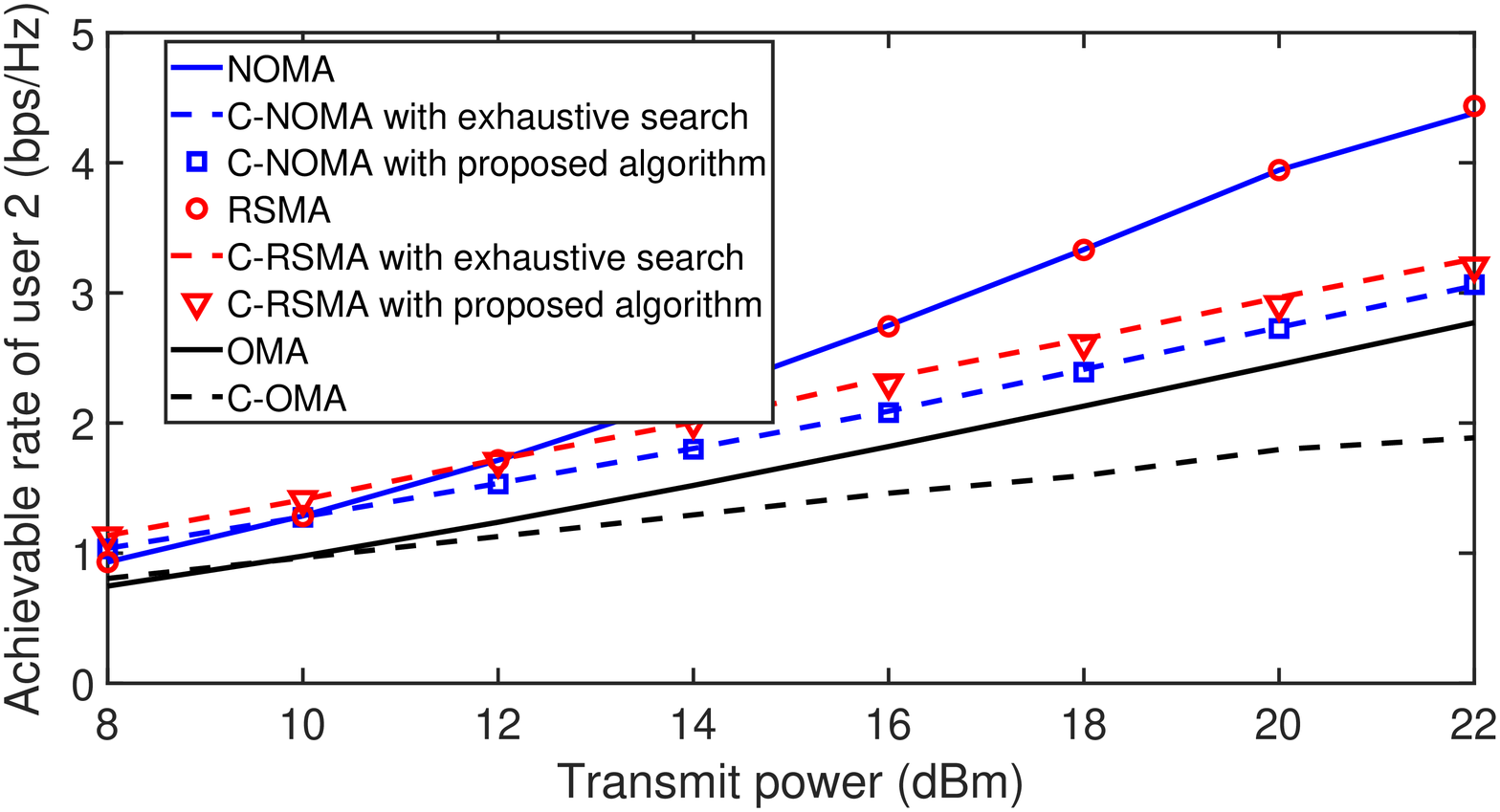}
\caption{$h_1-h_2=2 ~\mathrm{dB}$}
\label{fig:r2_d2_13}
\end{subfigure}
\caption{Achievable rate of user 2 versus transmit power utilizing the proposed algorithms and exhaustive search when the proportional fairness coefficient is $f=\frac{1}{3}$. Two cases $h_1-h_2=2 ~\mathrm{dB}$ and $h_1-h_2=10 ~\mathrm{dB}$ have been considered.}
\label{fig:r2_d210_13}
\end{figure}

Fig. \ref{fig:r2_p820_13} shows the achievable rate of user 2 versus $h_1-h_2$ when the proportional fairness coefficient is $f=\frac{1}{3}$. Fig. \ref{fig:r2_p8_13} shows the case that transmit power is low as $\bar{P}_1=\bar{P}_2=8~\mathrm{dBm}$. One can see that by increasing $h_1-h_2$, the superiority of cooperative schemes rather than non-cooperative schemes increases. This is due to the fact that when $h_1-h_2$ increases, the required rate of user 2 can not be satisfied, and cooperation with user 1 is helpful for user 2.  Fig. \ref{fig:r2_p20_13} indicates the case that transmit power is high as $\bar{P}_1=\bar{P}_2=20~\mathrm{dBm}$. We can see that when $h_1-h_2$ is small, due to their optimal detectors, non-cooperative schemes outperform cooperative schemes. However, by increasing $h_1-h_2$, cooperative schemes outperform non-cooperative ones. By comparing these two figures, we can say that for the superiority of cooperative schemes over their non-cooperative counterparts, more channel difference ($h_1-h_2$) is required at higher transmit powers. From these figures, we can also see that cooperative schemes are very robust to the channel variations, and they can provide almost a constant rate even if the channel power gain of one user goes to deep fading. Finally, one can see that by increasing $h_1-h_2$, the superiority of C-RSMA over C-NOMA decreases. Note that the superiority of C-RSMA over C-NOMA is at the cost of increased complexity in its transmission and detection schemes. \par

\begin{figure}[t]
\begin{subfigure}{0.5\textwidth}
\includegraphics[width=1\linewidth, height=4.3cm]{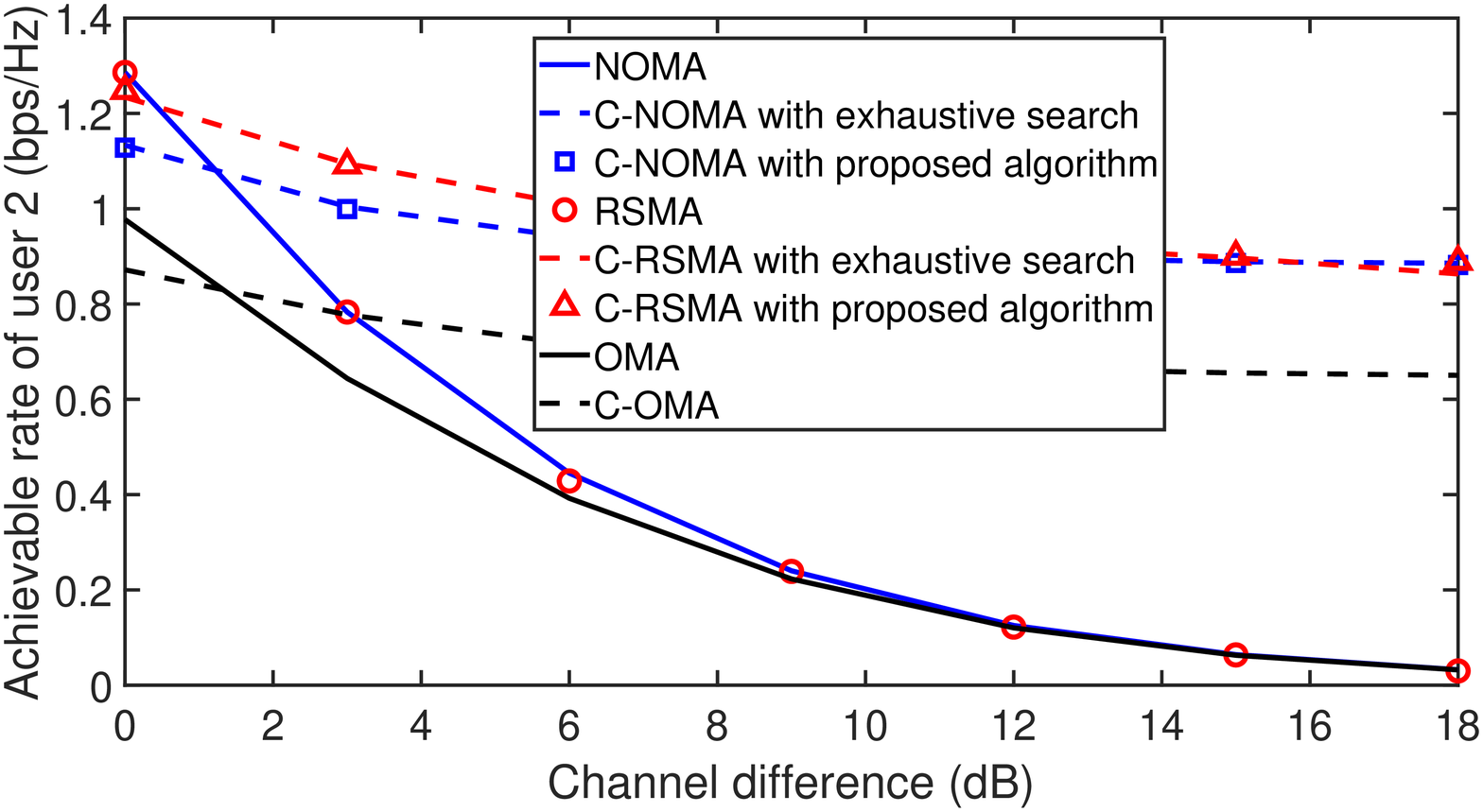} 
\caption{Low transmit power ($\bar{P}_1=\bar{P}_2=8~\mathrm{dBm}$)}
\label{fig:r2_p8_13}
\end{subfigure}
\begin{subfigure}{0.5\textwidth}
\includegraphics[width=1\linewidth, height=4.4cm]{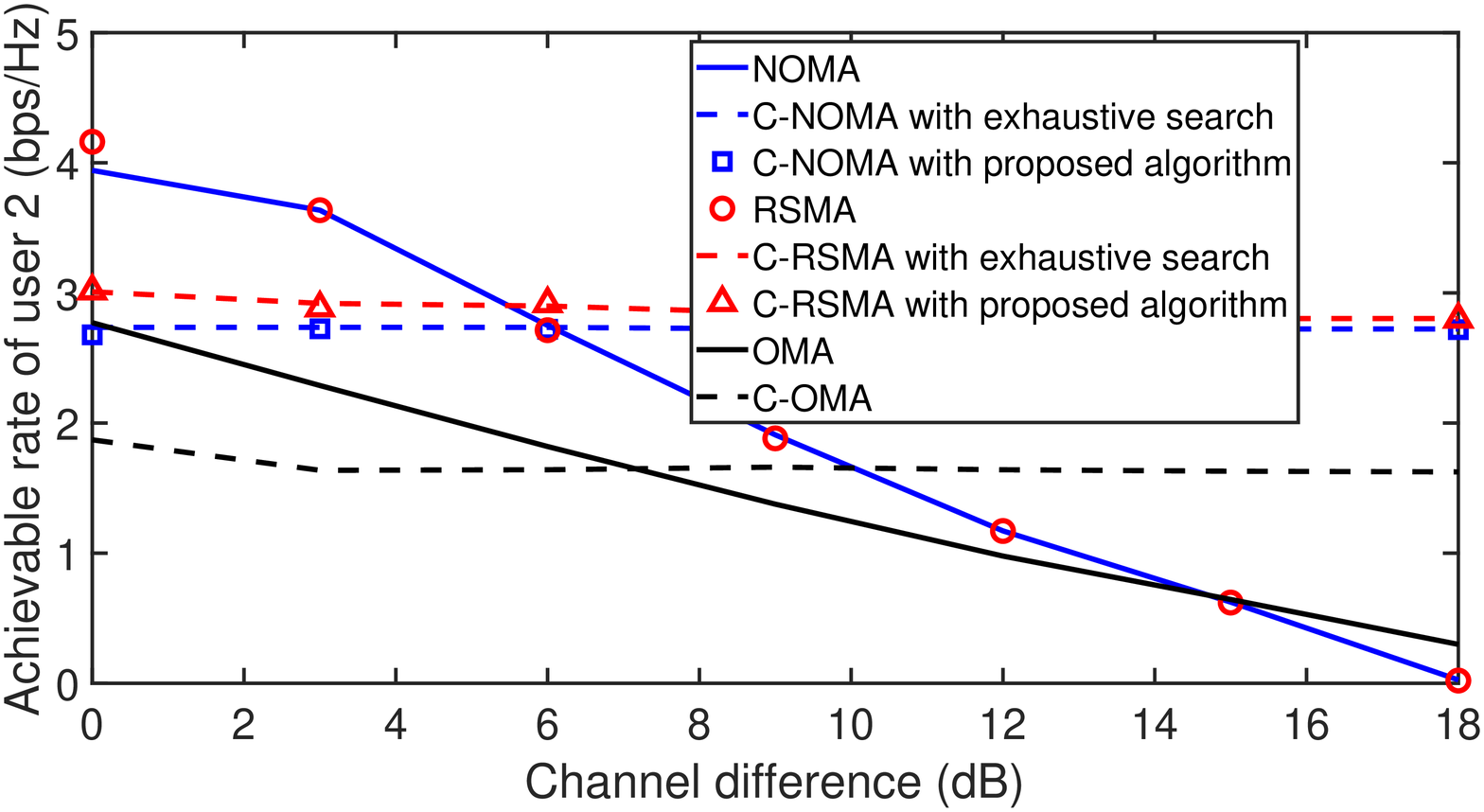}
\caption{High transmit power ($\bar{P}_1=\bar{P}_2=20~\mathrm{dBm}$)}
\label{fig:r2_p20_13}
\end{subfigure}
\caption{Achievable rate of user 2 versus  $h_1-h_2$ when $f=\frac{1}{3}$. Two different transmit powers $\bar{P}_1=\bar{P}_2=8~\mathrm{dBm}$ and $\bar{P}_1=\bar{P}_2=20~\mathrm{dBm}$ have been considered.}
\label{fig:r2_p820_13}
\end{figure}
 
 In Fig. \ref{amp_noise}, we can see the performance loss when the link between two users, i.e., $h_3$, is weak. In this figure, we see the achievable rate of user 2 versus the received SNR at users. For this figure, we assumed that the proportional fairness coefficient equals $f=1/3$, the transmit power at users is equal to $\bar{P}_1=\bar{P}_2=20~\mathrm{dBm}$, and $h_1-h_2=10~ \mathrm{dB}$. The ideal cases for C-OMA, C-NOMA, and C-RSMA schemes show the case that the impact of amplified noise due to AF relaying at users has been ignored. We can see that by increasing the received SNR at users, which is equivalent to a better link between two users, the achievable rate also increases. Indeed, when the received SNR at users is greater than 26 dB, the achievable rates of cooperative schemes are equal to their ideal counterparts, and hence the assumption of ignoring amplified noise is reasonable. For realizing this situation, we are just required to pair two users that the received SNR at their receivers is greater than 26 dB. Hence, we can see that $h_3$ is an important parameter for pairing users. Note that when the received SNR at users is less than 4 dB (6 dB), the non-cooperative RSMA (NOMA) has better performance rather than C-RSMA (C-NOMA), and hence the cooperation is meaningless in these cases. On the other hand, Fig. \ref{fig:r2_p820_13} shows that the achievable rates of the proposed schemes are robust to the variations of the $h_1-h_2$ and Hence, $h_1-h_2$ is not an important parameter for pairing users.\par

\begin{figure}[t]
\centering
\begin{minipage}{.48\linewidth}
  \includegraphics[width=\linewidth,height=4.3cm]{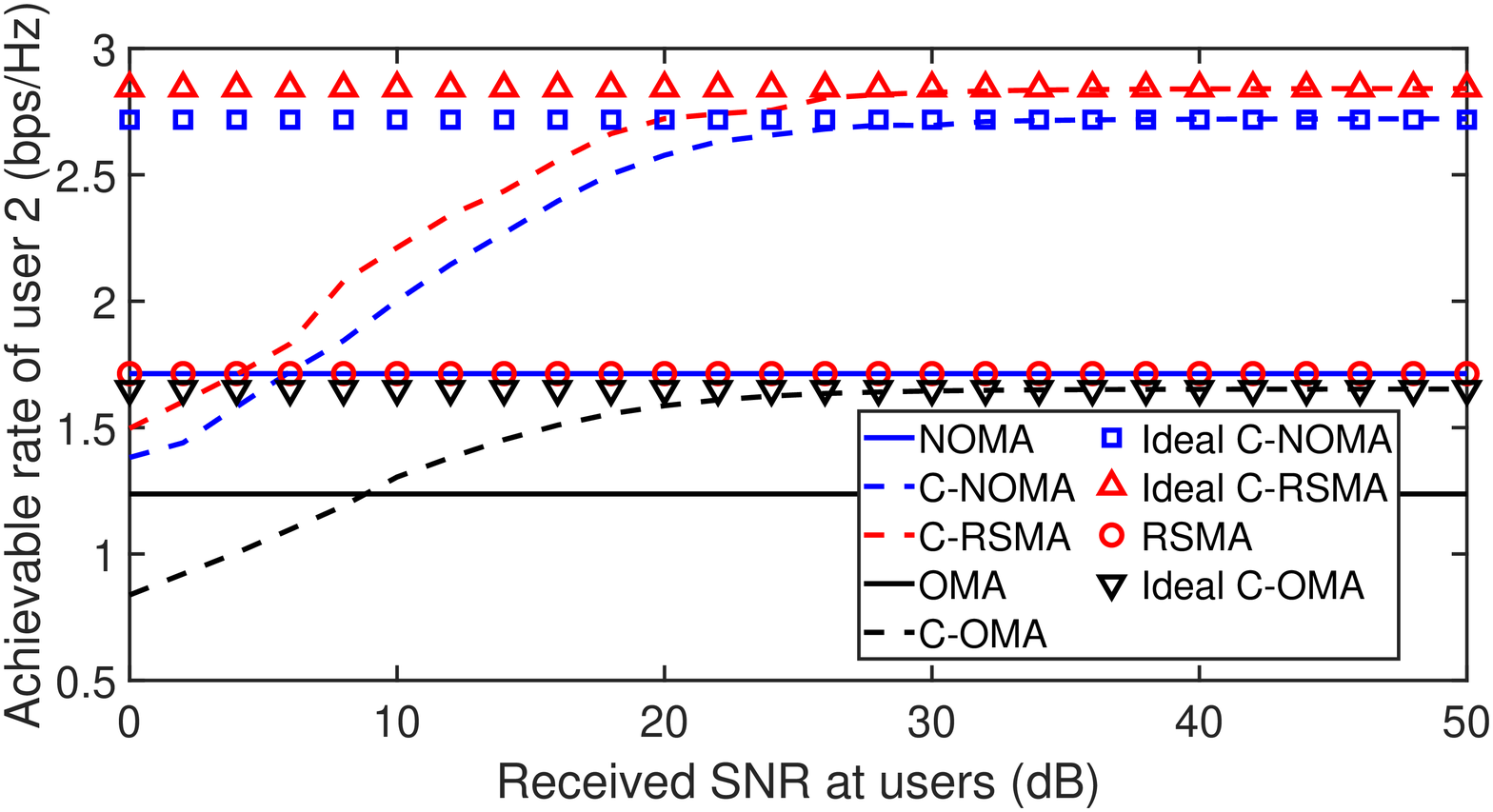}
  \captionof{figure}{Achievable rate of user 2 versus the received SNR at users. The ideal case for C-OMA, C-NOMA, and C-RSMA schemes shows the situation that the link between two users, i.e., $h_3$, is strong enough to ignore the noise at their receivers.} 
  \label{amp_noise}
\end{minipage}
 \hspace{.02\linewidth}
\begin{minipage}{.48\linewidth}
  \includegraphics[width=\linewidth,height=4.3cm]{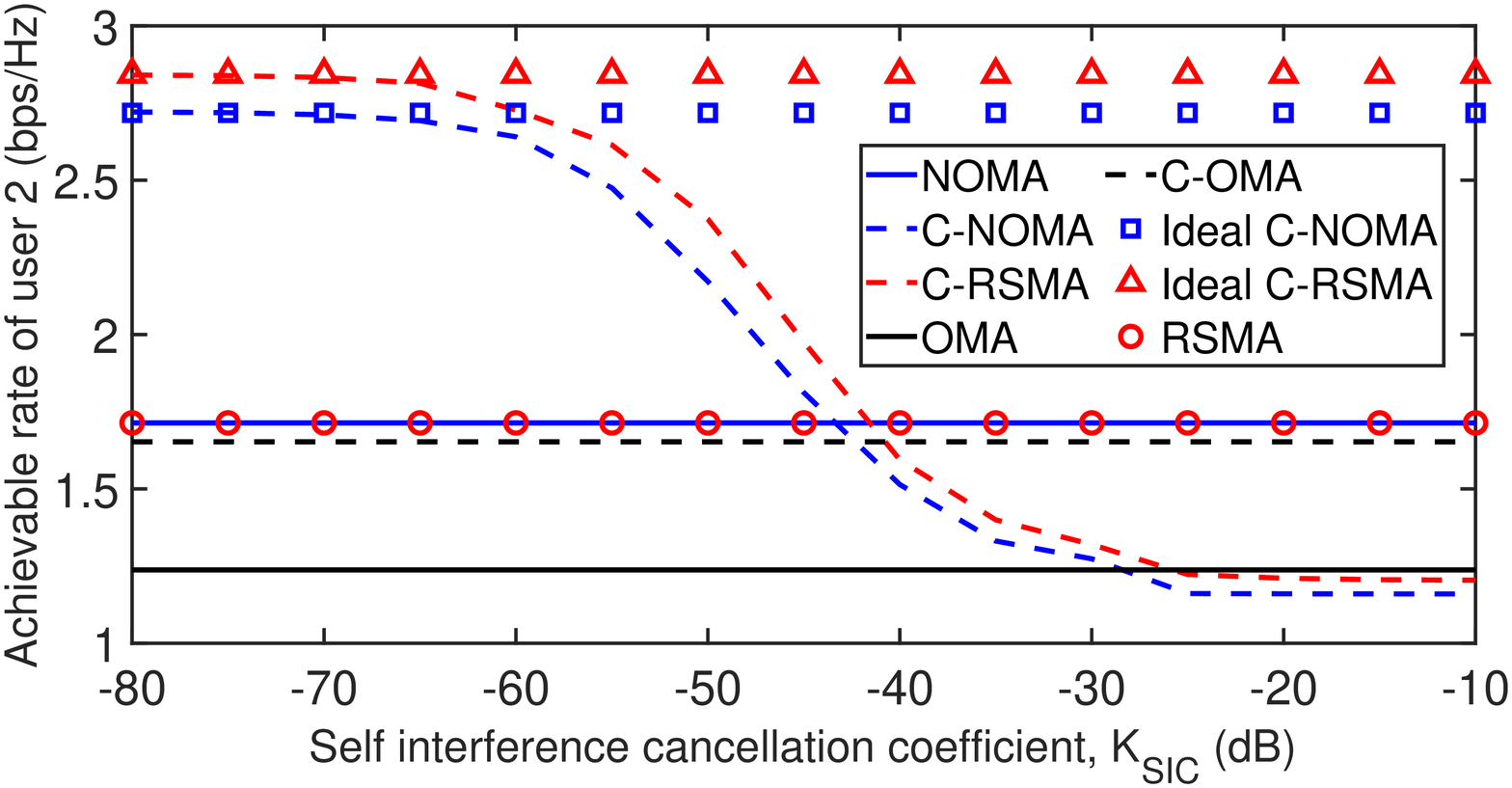}
  \captionof{figure}{Achievable rate of user 2 versus the self-interference cancellation coefficient at users, $K_{SIC}$ (dB). The ideal case shows the situation where the self-interference cancellation is performed perfectly at users, i.e., $K_{SIC}=0$.}
  \label{self_inter}
\end{minipage}
\end{figure}

The effect of residual self-interference in the performance of the proposed schemes has been depicted in Fig. \ref{self_inter}. In this figure, we see the achievable rate of user 2 versus the self-interference cancellation coefficient at users, $K_{SIC}$ (dB). Note that $K_{SIC}$ shows the amount of self-interference attenuation that can be performed at each user. The case $K_{SIC}=0$ shows the ideal case that there is no residual self-interference at users, and the case $K_{SIC}=1$ indicates the situation that no self-interference cancellation is performed at users.  For this figure, we assumed that the proportional fairness coefficient equals $f=1/3$, the transmit power at users is equal to $\bar{P}_1=\bar{P}_2=20~\mathrm{dBm}$, $h_1-h_2=10~dB$, and the received SNR at users is equal to 50 dB. We can see that by decreasing $K_{SIC}$, which is equivalent to a less residual self-interference at users, the achievable rate increases. Indeed, when $K_{SIC}$ is less than -65 dB, the achievable rates of cooperative schemes are equal to their ideal counterparts, and hence the assumption of ignoring self-interference is rational. Also, note that when $K_{SIC}$ is greater than -42 dB (-44 dB), non-cooperative RSMA (NOMA) has better performance rather than C-RSMA (C-NOMA).

In Fig. \ref{fig:Out}, we see the outage probability versus transmit power utilizing Monte-Carlo simulations for all schemes, and utilizing derived analytical expressions in Proposition \ref{proposition-outage} for C-NOMA and C-RSMA schemes. At these simulations, we assumed that $d_1=100~\mathrm{m}$, $d_2=120~\mathrm{m}$, and the required rates of users equal $R_{1}^{th}=R_{2}^{th}=0.5 ~\mathrm{bps/Hz}$ which means that $f=1$. Also, in the Monte-Carlo simulations, we depicted figures for $10^6$ realizations of the channels. In Fig. \ref{fig:Out_d10}, the difference between the average channel power gains of two users has been considered as  $\bar{h}_1-\bar{h}_2=10 ~\mathrm{dB}$. In this figure, we can see that simulation results match well with the analytical expressions derived in Proposition \ref{proposition-outage}. Also, matching with Corollary \ref{corollary1}, we can see that the diversity orders of the proposed cooperative schemes are two which is twice the diversity orders of non-cooperative schemes. We can see that the proposed C-NOMA and C-RSMA schemes have better outage performance rather than C-OMA and non-cooperative schemes. In Fig. \ref{fig:Out_d15}, by increasing $\bar{h}_1-\bar{h}_2=15 ~\mathrm{dB}$, we can see that the cooperative schemes are more robust  than non-cooperative schemes against this channel variation. Hence, the outage performance superiority of the cooperative schemes over non-cooperative schemes has been enhanced by increasing  $\bar{h}_1-\bar{h}_2$.

\begin{figure}[t]
\begin{subfigure}{0.5\textwidth}
\includegraphics[width=1\linewidth, height=4.3cm]{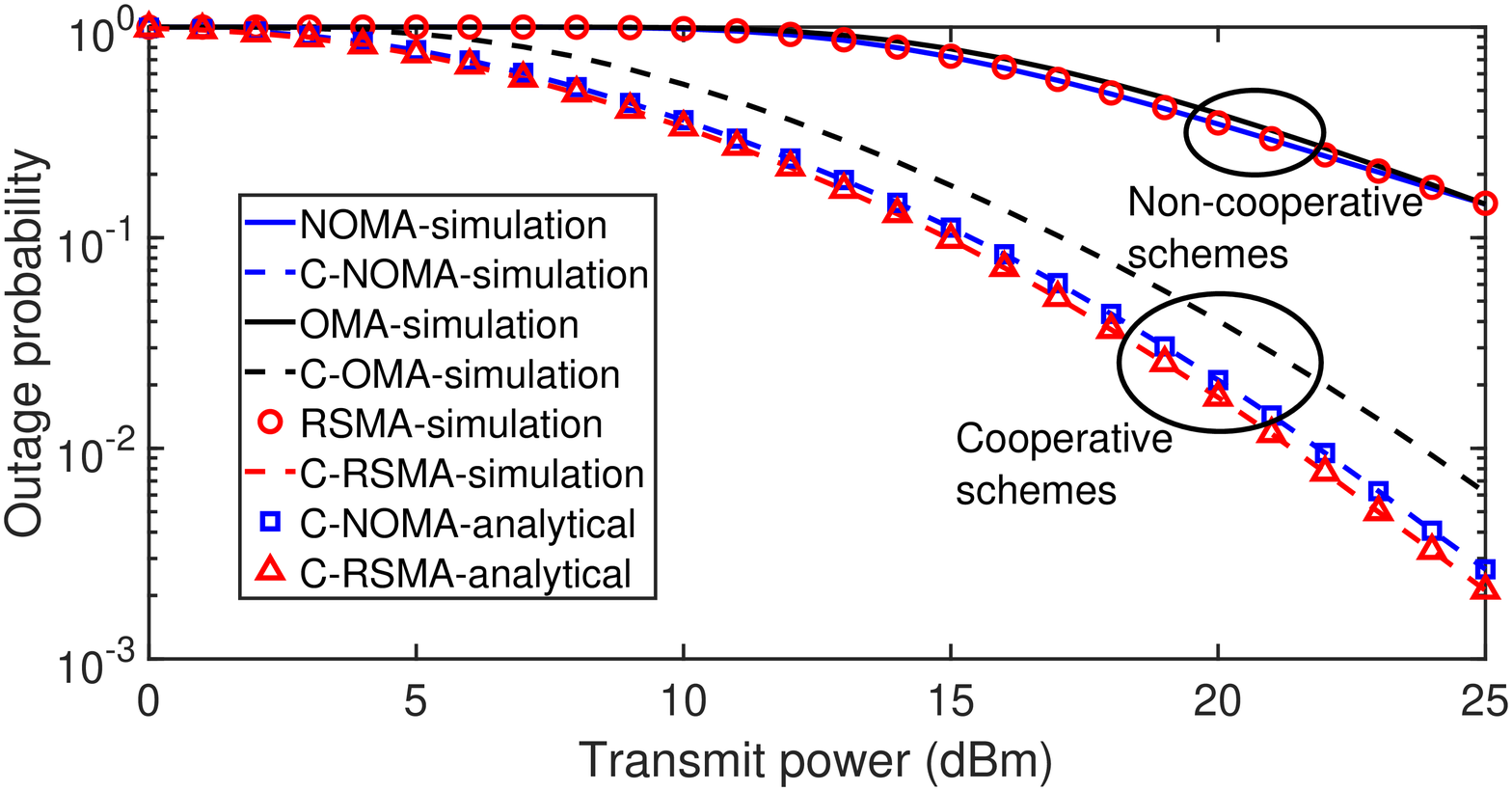} 
\caption{$\bar{h}_1-\bar{h}_2=10 ~\mathrm{dB}$}
\label{fig:Out_d10}
\end{subfigure}
\begin{subfigure}{0.5\textwidth}
\includegraphics[width=1\linewidth, height=4.3cm]{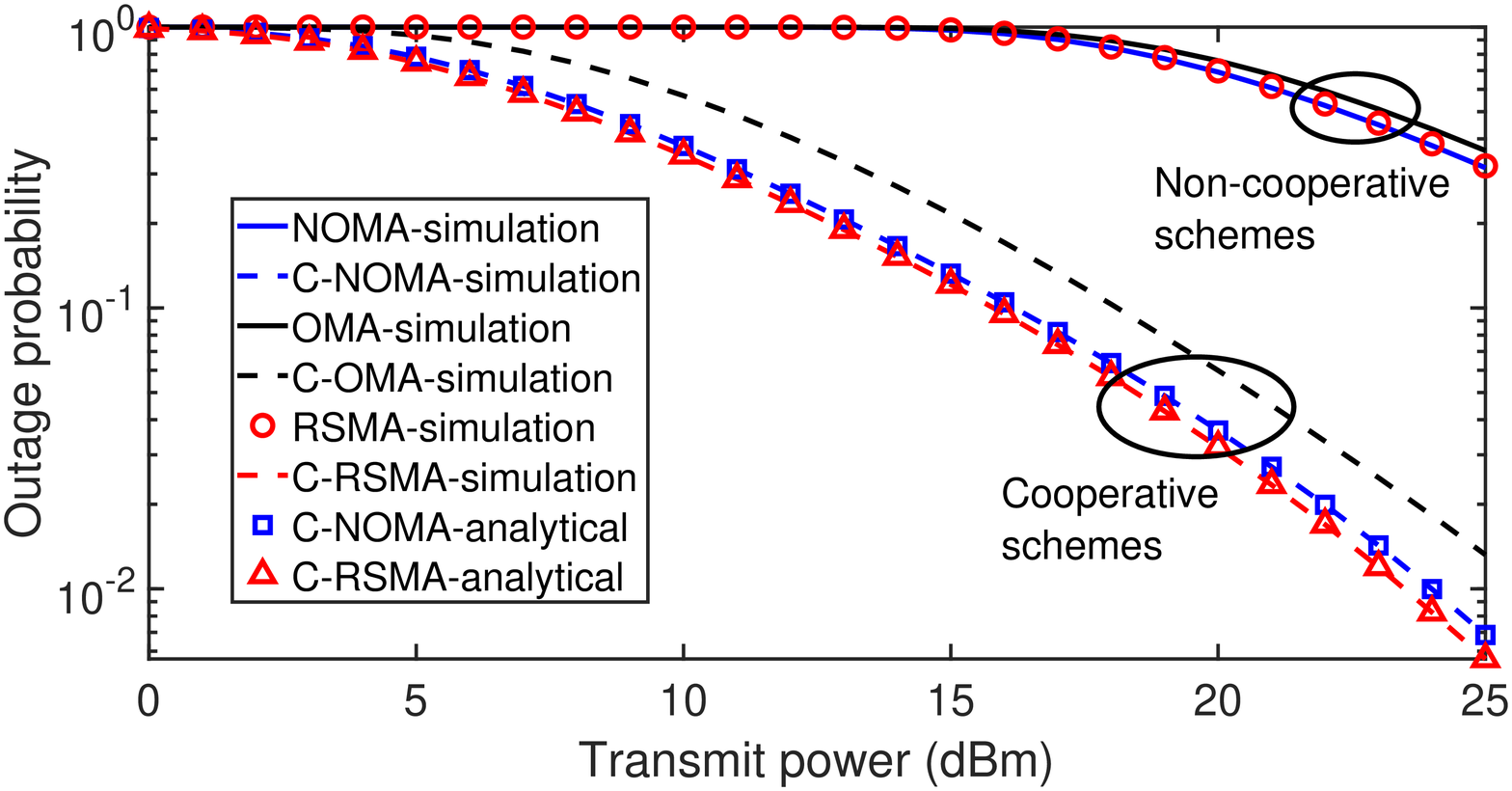}
\caption{$\bar{h}_1-\bar{h}_2=15 ~\mathrm{dB}$}
\label{fig:Out_d15}
\end{subfigure}
\caption{Outage probability versus transmit power utilizing Monte-Carlo simulations and derived analytical expressions when the required rates of users are $R_{1}^{th}=R_{2}^{th}=0.5 ~\mathrm{bps/Hz}$. Two cases $\bar{h}_1-\bar{h}_2=10 ~\mathrm{dB}$ and $\bar{h}_1-\bar{h}_2=15 ~\mathrm{dB}$ have been considered.}
\label{fig:Out}
\end{figure}

In Fig. \ref{fig:CDF}, we can see the CDF of the achievable rates of users, i.e., $R=R_1=R_2~\mathrm{bps/Hz}$, utilizing Monte-Carlo simulations. At these simulations, we assumed that $d_1=100~\mathrm{m}$, $d_2=120~\mathrm{m}$, and transmit powers of users equal $\bar{P}_1=\bar{P}_2=20~\mathrm{dBm}$. In Fig. \ref{fig:CDF_p20_d10}, the difference between the average channel power gains of two users has been considered as  $\bar{h}_1-\bar{h}_2=10 ~\mathrm{dB}$. From this figure, we can see that the cooperative schemes outperform non-cooperative schemes. Also, the proposed C-NOMA and C-RSMA schemes have better performance rather than the C-OMA scheme. In Fig. \ref{fig:CDF_p20_d15}, by increasing  $\bar{h}_1-\bar{h}_2=15 ~\mathrm{dB}$, we can see that the cooperative schemes are more robust  than non-cooperative schemes against this channel variation.  
\begin{figure}[t]
\begin{subfigure}{0.5\textwidth}
\includegraphics[width=1\linewidth, height=4.3cm]{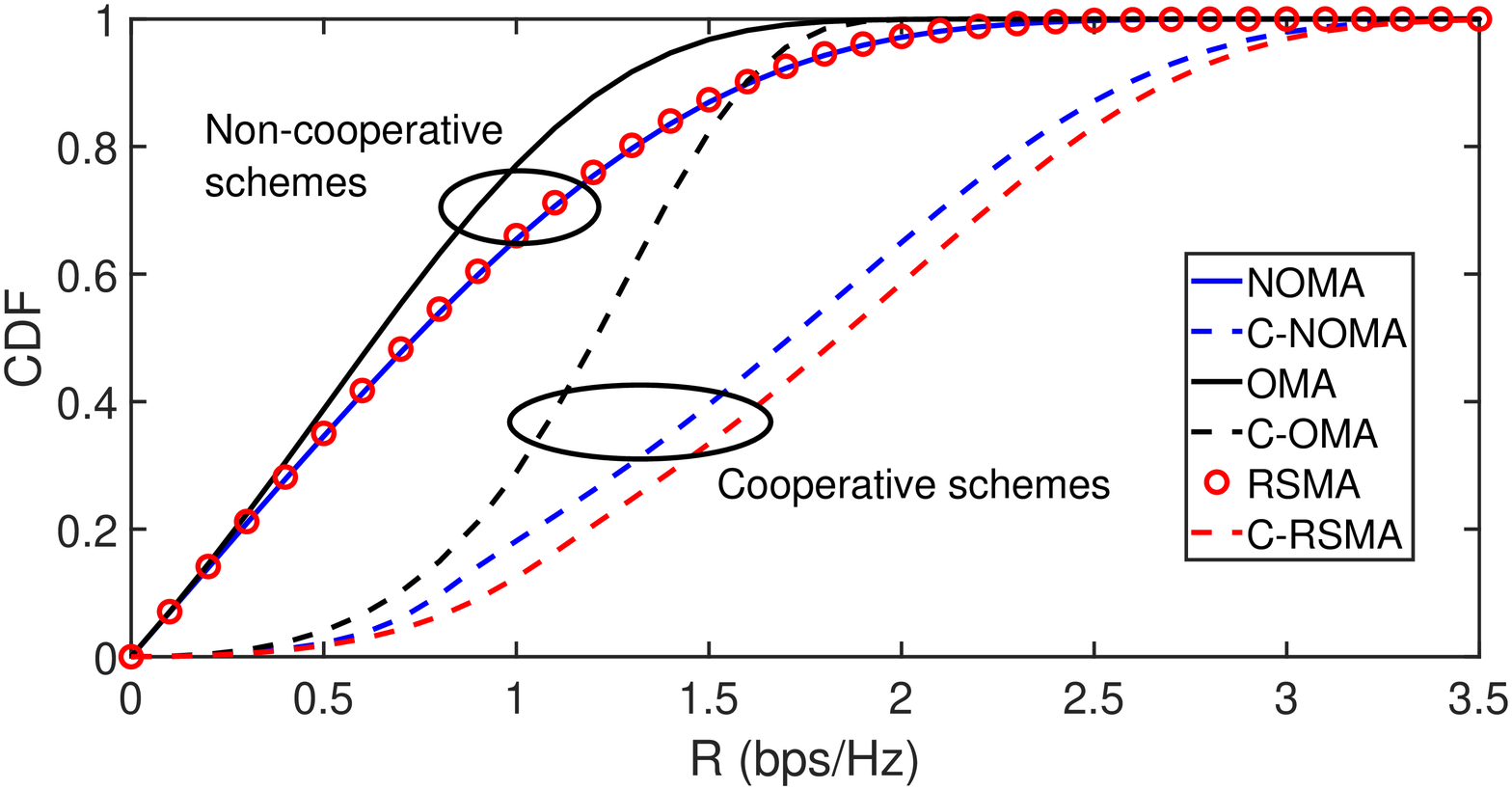} 
\caption{$\bar{h}_1-\bar{h}_2=10 ~\mathrm{dB}$}
\label{fig:CDF_p20_d10}
\end{subfigure}
\begin{subfigure}{0.5\textwidth}
\includegraphics[width=1\linewidth, height=4.3cm]{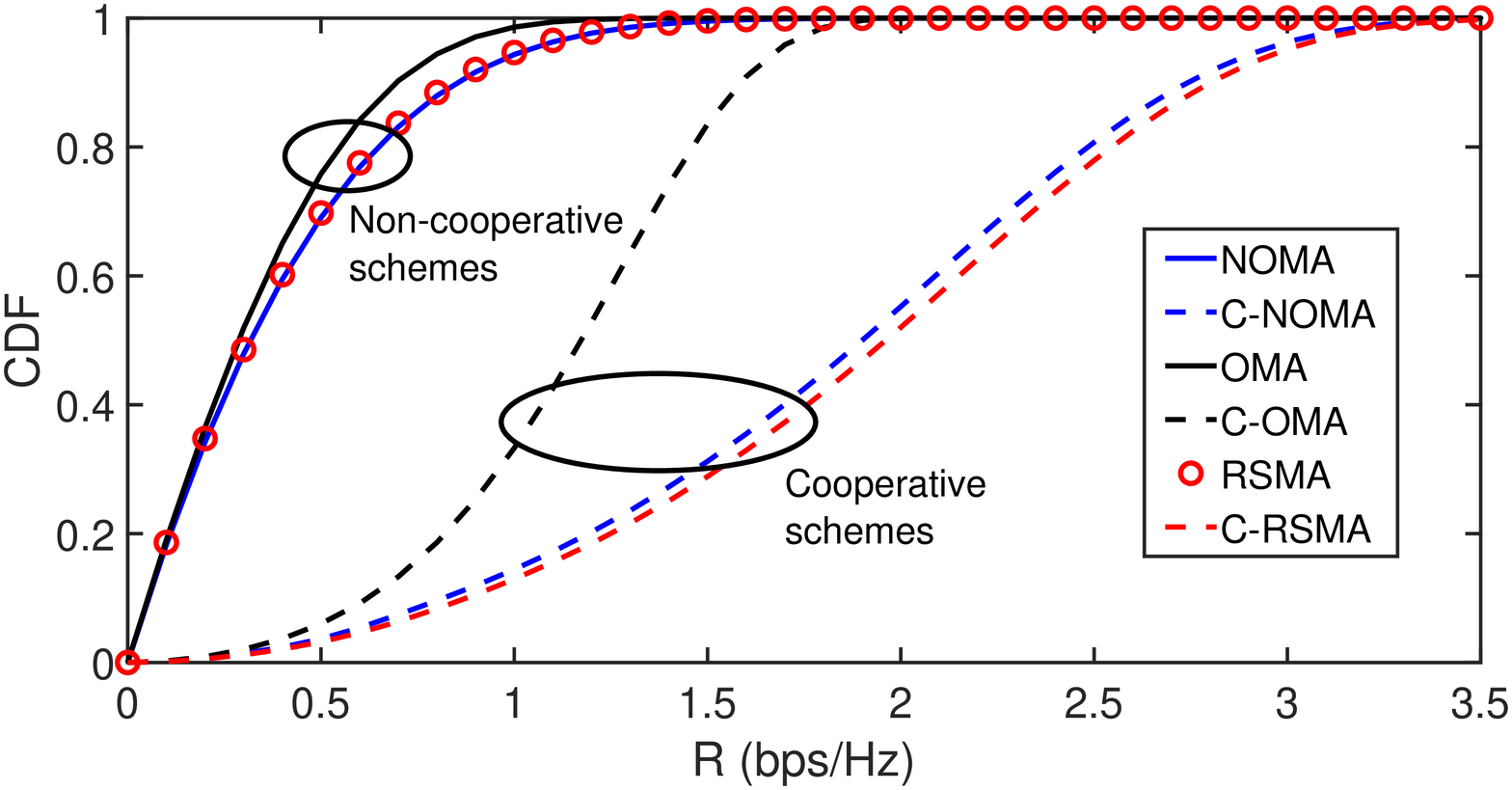}
\caption{$\bar{h}_1-\bar{h}_2=15 ~\mathrm{dB}$}
\label{fig:CDF_p20_d15}
\end{subfigure}
\caption{CDF of the achievable rates of users, i.e.,  $R=R_1=R_2~\mathrm{bps/Hz}$ when transmit powers are $\bar{P}_1=\bar{P}_2=20~\mathrm{dBm}$. Two cases $\bar{h}_1-\bar{h}_2=10 ~\mathrm{dB}$ and $\bar{h}_1-\bar{h}_2=15 ~\mathrm{dB}$ have been considered.}
\label{fig:CDF}
\end{figure}

\section{Conclusion}
 In this paper, two novel transmission and detection schemes were proposed for uplink communication. We derived the achievable rates corresponding to these cooperative schemes, and formulated two optimization problems to maximize the minimum rate of two users by considering the proportional fairness coefficient. We proposed two power allocation algorithms based on SCA and GP to solve these non-convex problems. We proved that the proposed efficient algorithms are guaranteed to converge. Next, we derived the asymptotic outage probability of the proposed C-NOMA and C-RSMA schemes, and we proved that the proposed cooperative schemes achieve a diversity order of two. Finally, simulation results showed that C-RSMA has better performance rather than C-NOMA, and C-NOMA has a much better performance compared with C-OMA. These superiorities are at the cost of increased complexity at the transmitter and receiver of C-NOMA and C-RSMA schemes. Also, simulation results showed that due to their higher diversity order, cooperative schemes are very robust to the channel variations compared with non-cooperative schemes.
 
%
\appendices
\section{Proof of Lemma 1}\label{lemma 1}
In order to prove that $f$ is convex, we calculate its first and second order derivatives as 
\small
\begin{equation}
\nabla f=
\begin{bmatrix}
\frac{\partial f}{\partial x}\\
\frac{\partial f}{\partial y}
\end{bmatrix}
=
\begin{bmatrix}
\frac{-a}{x^2+ax+\frac{bx^2}{y}}\\
\frac{-b}{y^2+by+\frac{ay^2}{x}}
\end{bmatrix},\nabla^{2} f=
\begin{bmatrix}
\frac{\partial f  ^{2}}{\partial  x ^{2}} & \frac{\partial f  ^{2}}{\partial  y \partial  x} \\
\frac{\partial f^{2}}{\partial  x \partial y} & \frac{\partial f  ^{2}}{\partial  y ^{2}}
\end{bmatrix}
=\frac{1}{u^{2}}
\begin{bmatrix}
\frac{2ay^2}{x}+\frac{a^2y^2}{x^2}+\frac{2aby}{x}& -ab\\
-ab&\frac{2bx^2}{y}+\frac{b^2x^2}{y^2}+\frac{2abx}{y}
\end{bmatrix},
\end{equation}
\normalsize
where $u=xy+ay+bx$. In order to show that function $f$ is a convex function, we must prove that the Hermitian matrix $\nabla^{2} f$ is a positive definite matrix. For this end, we utilize Sylvester's criterion. According to this criterion, the Hermitian matrix $\nabla^{2} f$ is positive definite if and only if the upper-left 1-by-1 corner of $\nabla^{2} f$, and
the upper left 2-by-2 corner of $\nabla^{2} f$ have a positive determinant. We can see that both of these criteria are true for Matrix $\nabla^{2} f$.

\section{Proof of Proposition \ref{prop-cnoma-opt} }\label{appendix-cnoma}
First, we perform variable changes $x_1=\frac{\gamma_2p_2^1}{\gamma_1p_1^1}+\frac{1}{\gamma_1p_1^1}$, $y_1=\frac{\gamma_1p_1^2}{\gamma_2p_2^2}+\frac{1}{\gamma_2p_2^2}$, $x_2=\frac{1}{\gamma_1p_1^2}$, and $y_2=\frac{1}{\gamma_2p_2^1}$. Then, we rewrite the achievable rates of C-NOMA scheme in (\ref{r1_cnoma}) and (\ref{r2_cnoma}) with these new variables as
\small
\begin{equation}\label{r1_cnoma_lb}
  R_1^{\mathrm{C-NOMA}}=\frac{1}{2}\log_2(1+\frac{1}{x_1}+\frac{1}{y_1}), R_2^{\mathrm{C-NOMA}}=\frac{1}{2}\log_2(1+\frac{1}{x_2}+\frac{1}{y_2}).  
\end{equation}
\normalsize
According to Lemma \ref{lemma1}, these rates are convex with respect to $x_1$, $x_2$, $y_1$, and $y_2$, though they are not convex with respect to power allocation coefficients. In the SCA method, we solve the original non-convex problem by solving approximated convex problems of the original problem around one initial point iteratively. We know that the first-order Taylor series expansion of a convex function $f(z)$ provides a global lower-bound for that function, i.e., $f(z)\geq f(z_0)+\nabla f(z_0)^T(z-z_0)$.
Considering this point that we have a convex approximation of the rates with respect to $x_1$, $x_2$, $y_1$, and $y_2$, we can approximate them at each iteration $l+1$ by their first-order Taylor series expansion around the solution of previous iteration $l$. Therefore, the achievable rate of user $i$ at the $(l+1)^{\text{th}}$ iteration can be approximated by
\small
\begin{equation}\label{taylor}
    R_{\text{lb},i}^{l+1,\mathrm{C-NOMA}}=\frac{1}{2}R_{\text{lb},i}^{l,\mathrm{C-NOMA}}+\frac{1}{2}d_{i1}^l(x_i^{l+1}-x_i^l)+\frac{1}{2}d_{i2}^l(y_i^{l+1}-y_i^l),
\end{equation}
\normalsize
 where $d_{i1}^l$ and $d_{i2}^l$  are partial derivatives as $d_{i1}^l=\frac{-1}{\ln{2}\big((x_i^l)^2+x_i^l+(\frac{x_i^l}{y_i^l})\big)}$ and $d_{i2}^l=\frac{-1}{\ln{2}\big((y_i^l)^2+y_i^l+(\frac{y_i^l}{x_i^l})\big)}$ (for $i\in\{1, 2\}$). By replacing original power allocation variables instead of $x_1$, $x_2$, $y_1$, and $y_2$ in (\ref{taylor}), we obtain (\ref{rate-cnoma_lb}). It is clear that $R_{\text{lb},i}^{l+1,\mathrm{C-NOMA}}$ in (\ref{taylor}) is an affine function with respect to $x_i^{l+1}$ and $y_i^{l+1}$. Finally, note that due to applying first-order Taylor expansion for a convex function, we have $R_{i}^{l+1,\mathrm{C-NOMA}}>R_{\text{lb},i}^{l+1,\mathrm{C-NOMA}}$, and due to the  maximization of approximated function at each iteration $l$, we have $R_{\text{lb},i}^{l+1,\mathrm{C-NOMA}}>R_{\text{lb},i}^{l,\mathrm{C-NOMA}}$. Considering this point that Taylor expansion of a function around an initial point is exactly equal to the value of the original function at that point, i.e., $R_{\text{lb},i}^{l,\mathrm{C-NOMA}}=R_{i}^{l,\mathrm{C-NOMA}}$, we can conclude that $R_{i}^{l+1,\mathrm{C-NOMA}}>R_{i}^{l,\mathrm{C-NOMA}}$. Hence, the rates $R_{i}^{l,\mathrm{C-NOMA}}$ are non-decreasing with $l$ and the proof is completed.

\section{Proof of Proposition \ref{proposition-outage} }\label{appendix-outage}
The achievable rates of users at C-NOMA and C-RSMA schemes have been derived in Proposition \ref{proposition-achievable-rate-C-NOMA} and \ref{proposition-achievable-rate-C-RSMA}, respectively. Note that the allocated powers for these schemes in P1 and P2 are so that at the optimal solution we have $R_1^{\mathrm{C-NOMA}}=fR_2^{\mathrm{C-NOMA}}$, and $R_1^{\mathrm{C-RSMA}}=fR_2^{\mathrm{C-RSMA}}$. Hence, the outage probability of these schemes can be written as $P_{out}^{C-NOMA}=\mathrm{Pr}( R_{2}^{C-NOMA}<R_{2}^{th})$, and $P_{out}^{C-RSMA}=\mathrm{Pr}( R_{2}^{C-RSMA}<R_{2}^{th})$. From (\ref{r2_cnoma}), the outage probability of the C-NOMA scheme is given by
\small
\begin{equation}
    P_{out}^{C-NOMA}=\mathrm{Pr}(\frac{1}{2}\log_2(1+\gamma_1p_1^2+\gamma_2p_2^1)<R_{2}^{th}),
\end{equation}
\normalsize
where $\gamma_i=\frac{h_i\bar{P}_{i}}{\sigma_{BS}^2}$ for $i\in\{1, 2\}$ is an exponential random variable with the mean value of $\Omega_i=\frac{\bar{P}_{i}\beta_0}{\sigma_{BS}^2d_i^\alpha
}$ indicating the SNR of each user. Now, by utilizing Lemma \ref{lemma2} we have
\small
\begin{equation}
    P_{out}^{C-NOMA}=1-\exp{(-\frac{2^{2R_{2}^{th}}-1}{\Omega_2p_2^1})}-\frac{\exp{(-\frac{2^{2R_{2}^{th}}-1}{\Omega_1p_1^2})}-\exp{(-\frac{2^{2R_{2}^{th}}-1}{\Omega_2p_2^1}})}{1-\frac{\Omega_2p_2^1}{\Omega_1p_1^2}}.
\end{equation}
\normalsize
Next, in order to derive asymptotic outage probability at high SNR, we utilize the approximation $\exp (x)\approx 1+x+\frac{x^2}{2}$ when $x\rightarrow 0$, and after some manipulation we derive (\ref{out_cnoma}). For C-RSMA scheme, it's clear that the second logarithm in (\ref{r2_crsma}), i.e., $\log_2(1+p_2^1p_{22}\gamma_2+p_1^2p_{22}\gamma_1)$, represents $R_{2}^{C-RSMA}$ at high SNR regime. Hence, applying the same procedure with the C-NOMA scheme, the asymptotic outage probability of the C-RSMA scheme can be derived as (\ref{out_crsma}).




\ifCLASSOPTIONcaptionsoff
  \newpage
\fi



\bibliographystyle{IEEEtran}
\bibliography{myref}
%

%








\end{document}